\documentclass[prodmode,acmtocl]{acmsmall}

\usepackage[ruled]{algorithm2e}

\SetAlFnt{\small}
\SetAlCapFnt{\small}
\SetAlCapNameFnt{\small}
\SetAlCapHSkip{0pt}
\IncMargin{-\parindent}

\usepackage{multirow}
\usepackage{listings}
\usepackage{textcomp}
\usepackage{amssymb}
\usepackage{url}
\usepackage{amsmath}
\usepackage{latexsym}
\usepackage{epsfig}
\usepackage{color}
\usepackage{xspace}
\usepackage{supertabular}
\usepackage{float}

\lstnewenvironment{boxlisting}{\lstset{gobble=4,frame=trbl,belowskip=1mm,aboveskip=4mm,xrightmargin=7cm,xleftmargin=0.6cm}}{}

\lstMakeShortInline@

\newcommand{\goes}{\rightarrow}
\newcommand{\st}{\;|\;}

\newcommand{\hard}{H}
\newcommand{\soft}{S}

\newcommand{\optimizeQFLIA}                {\textit{optimize\_QF\_LIA}\xspace}
\newcommand{\optimizeQFLIAMaxSMT}          {\textit{optimize\_QF\_LIA\_Max\_SMT}\xspace}
\newcommand{\optimizeQFLIAMaxSMTthreshold} {\textit{optimize\_QF\_LIA\_Max\_SMT\_threshold}\xspace}
\newcommand{\optimizeQFLIAOMT}             {\textit{optimize\_QF\_LIA\_OMT}\xspace}

\newcommand{\relaxDomainsMinModels}        {\textit{new\_domains\_min\_models}\xspace}
\newcommand{\relaxDomainsMinModelsNonInc}   {\textit{new\_domains\_min\_models\_non\_inc}\xspace}
\newcommand{\artificialBounds}             {\textit{artificial\_bounds}\xspace}
\newcommand{\linearize}                    {\textit{linearize}\xspace}
\newcommand{\relaxDomainsCores}            {\textit{new\_domains\_cores}\xspace}
\newcommand{\update}                       {\textit{update}\xspace}
\newcommand{\updateNonInc}                 {\textit{update\_non\_inc}\xspace}
\newcommand{\solveSMTQFNIAMinModels}       {\textit{solve\_SMT\_QF\_NIA\_min\_models}\xspace}
\newcommand{\solveSMTQFNIACores}           {\textit{solve\_SMT\_QF\_NIA\_cores}\xspace}
\newcommand{\solveSMTQFLIA}                {\textit{solve\_SMT\_QF\_LIA}\xspace}
\newcommand{\solveOMTQFLIA}                {\textit{solve\_OMT\_QF\_LIA}\xspace}

\newcommand{\solveMaxSMTQFLIA}             {\textit{solve\_Max\_SMT\_QF\_LIA}\xspace}
\newcommand{\solveMaxSMTQFNIA}             {\textit{solve\_Max\_SMT\_QF\_NIA}\xspace}

\renewcommand{\b}[1]{{\textbf{#1}}}

\newcommand{\mr}[3]  {{\multirow{#1}{#2}{#3}}}

\newcommand{\Weight}               {\Omega}
\newcommand{\weight}               {\omega}

\renewcommand{\l}               {L}
\renewcommand{\u}               {U}

\newcommand{\msc}    {\mathit{max\_soft\_cost}}
\newcommand{\undef}  {\bot}
\newcommand{\bsf}    {\mathit{best\_so\_far}}

\newcommand{\bounds} {\mathit{B}}
\newcommand{\nonlinear} {\mathit{N}}
\newcommand{\status} {\mathit{ST}}
\newcommand{\core}   {\mathit{UC}}
\newcommand{\model}  {\mathit{M}}
\newcommand{\cost}   {\mathit{cost}}
\newcommand{\Sat}    {\mathsf{SAT}}
\newcommand{\Unsat}  {\mathsf{UNSAT}}
\newcommand{\Unknown}{\mathsf{UNKNOWN}}

\newcommand{\Optimal}{\mathsf{OPT}}

\newcommand{\sat}   {SAT\xspace}           
\newcommand{\maxsat}{Max-SAT\xspace}       \newcommand{\MAXSAT}{Max-SAT}
\newcommand{\smt}   {SMT\xspace}           \newcommand{\SMT}   {SMT}
\newcommand{\maxsmt}{Max-SMT\xspace}       \newcommand{\MAXSMT}{Max-SMT}
\newcommand{\omt}   {OMT\xspace}           \newcommand{\OMT}   {OMT}

\newcommand{\smtrat}    {{\sf SMT-RAT}\xspace}

\newcommand{\limplies}{\rightarrow}

\renewcommand{\comment}[1]{}
\newcommand{\integers}{\mathbb{Z}}
\newcommand{\reals}{\mathbb{R}}

\newcommand{\ea}{$\exists \integers\, \forall \reals$-NIRA\xspace}

\newcommand{\rasat}[1]        {\textsf{raSAT#1}\xspace}
\newcommand{\prob}            {\textsf{ProB}\xspace}
\newcommand{\cvc}[1]          {\textsf{CVC#1}\xspace}
\newcommand{\z}[1]            {\textsf{z#1}\xspace}
\newcommand{\yices}[1]        {\textsf{yices-#1}\xspace}
\newcommand{\aproveNia}       {\textsf{AProVE-NIA}\xspace}
\newcommand{\barcelogicCores} {\textsf{bcl-cores}\xspace}
\newcommand{\barcelogicMaxsmt}{\textsf{bcl-maxsmt}\xspace}
\newcommand{\barcelogicOmt}   {\textsf{bcl-omt}\xspace}
\newcommand{\barcelogicJmpNoCores} {\textsf{bcl-ninc}\xspace}
\newcommand{\barcelogicJmpCores}   {\textsf{bcl-ninc-cores}\xspace}
\newcommand{\tool}[1]         {\textsf{#1}\xspace}

\def\defemb#1#2{\expandafter\def\csname #1\endcsname
                   {\relax\ifmmode #2\else\hbox{$#2$}\fi}}
\defemb{cH}{{\cal H}}
\defemb{cF}{{\cal F}}
\defemb{cB}{{\cal B}}
\defemb{cC}{{\cal C}}
\defemb{cE}{{\cal E}}
\defemb{cL}{{\cal L}}
\defemb{cN}{{\cal N}}
\defemb{cP}{{\cal P}}
\defemb{cS}{{\cal S}}
\defemb{cT}{{\cal T}}
\defemb{cU}{{\cal U}}
\defemb{cI}{{\cal I}}
\defemb{cD}{{\cal D}}

\issn{1529-3785}

\begin{document}

\markboth{Cristina Borralleras et al.}{Incomplete \smt Techniques for Solving Non-Linear Formulas over the Integers}

\title{Incomplete \smt Techniques for Solving Non-Linear Formulas\\ over the Integers}
\author{CRISTINA BORRALLERAS
\affil{Universitat de Vic - Universitat Central de Catalunya}
DANIEL LARRAZ,
ALBERT OLIVERAS,\\
ENRIC RODR\' IGUEZ-CARBONELL and
ALBERT RUBIO
\affil{Universitat Polit\`ecnica de Catalunya}}

\begin{abstract}
  We present new methods for solving the Satisfiability
  Modulo Theories problem over the theory of Quantifier-Free
  Non-linear Integer Arithmetic, SMT(QF-NIA), which consists in
  deciding the satisfiability of ground formulas with integer
  polynomial constraints. Following previous work, we propose to solve
  SMT(QF-NIA) instances by reducing them to linear arithmetic:
  non-linear monomials are linearized by abstracting them with fresh
  variables and by performing case splitting on integer variables with
  finite domain. For variables that do not have a finite domain, we
  can artificially introduce one by imposing a lower and an upper
  bound, and iteratively enlarge it until a solution is found (or the
  procedure times out).

  The key for the success of the approach is to determine, at each
  iteration, which domains have to be enlarged. Previously,
  unsatisfiable cores were used to identify the domains to be changed,
  but no clue was obtained as to how large the new domains should be.
  Here we explain two novel ways to guide this process by analyzing
  solutions to optimization problems: (i) to minimize the number of
  violated artificial domain bounds, solved via a Max-SMT solver, and
  (ii) to minimize the distance with respect to the artificial
  domains, solved via an Optimization Modulo Theories (OMT) solver.
  Using this SMT-based optimization technology allows smoothly
  extending the method to also solve Max-SMT problems over non-linear
  integer arithmetic. Finally we leverage the resulting
  Max-SMT(QF-NIA) techniques to solve $\exists \forall$ formulas in a
  fragment of quantified non-linear arithmetic that appears commonly
  in verification and synthesis applications.
\end{abstract}

%


%
%


\acmformat{
Cristina Borralleras, Daniel Larraz, Albert Oliveras, Enric Rodr\'{\i}guez-Carbonell, and
Albert Rubio, 2016.
Incomplete SMT techniques for solving non-linear formulas over the integers.}

\begin{bottomstuff}
  All authors partially supported by Spanish Ministerio de
  Econom\'{\i}a y Competitividad under grants
  TIN2013-45732-C4-3-P (MINECO) and TIN2015-69175-C4-3-R (MINECO/FEDER).
  Albert Oliveras is partially funded by the European Research Council (ERC) under the European
  Union's Horizon 2020 research and innovation programme (grant
  agreement ERC-2014-CoG 648276 AUTAR).

  Author's addresses:
  C. Borralleras,
  Departament Tecnologies Digitals i de la Informaci\'o, Escola Polit\`ecnica Superior, Universitat de Vic - Universitat Central de Catalunya;
  Daniel Larraz,
  Albert Oliveras,
  Enric Rodr\'{\i}guez-Carbonell and Albert Rubio, Departament of Computer Science, Universitat Polit\`ecnica de Catalunya.
\end{bottomstuff}

\maketitle

\section{Introduction}
\label{sec:introduccion}

Polynomial constraints are pervasive in computer science. They appear
naturally in countless areas, ranging from the analysis, verification
and synthesis of software and hybrid systems
\cite{Colonetal2003CAV,SankaranarayananSM04,SankaranarayananSM08,DBLP:conf/cade/PlatzerQR09}
to, e.g., game theory \cite{Beyene:2014:CAS:2535838.2535860}. In all
these cases, it is crucial to have efficient automatic solvers that,
given a formula involving polynomial constraints with integer or real
variables, either return a solution to the formula or report that
there is none.

Therefore, it is of no surprise that solving this sort of non-linear
formulas has attracted wide attention over the years. A milestone
result of Tarski's \cite{Tarski51} is a constructive proof that the
problem is decidable for the first-order theory of real closed fields,
in particular for the real numbers. Unfortunately the algorithm in the
proof has non-elementary complexity, i.e., its cost cannot be bounded
by any finite tower of exponentials, and is thus essentially useless
from a practical point of view. For this reason, for solving
polynomial constraints in $\reals$, computer algebra has traditionally
relied on the more workable approach of cylindrical algebraic
decomposition (CAD)
\cite{DBLP:conf/automata/Collins75,BasPolRoy_AlgInRealAlgGeom_2003}.
Still, its applicability is hampered by its doubly exponential
complexity.

Due to the limitations of the existing techniques, further research
has been carried out in polynomial constraint solving, spurred in the
last decade by the irruption of \sat and its extensions
\cite{HandbookOfSAT2009,NieuwenhuisOliverasTinelli2006JACM}. Thus,
several techniques have emerged in this period which leverage the
efficiency and automaticity of this new technology. E.g., for solving
polynomial constraints in $\reals$, interval constraint propagation
has been integrated with \sat and \smt engines
\cite{DBLP:journals/jsat/FranzleHTRS07,DBLP:conf/fmcad/GaoGIGSC10,DBLP:journals/entcs/KhanhO12}.
Other works pre-process non-linear formulas before passing them to an
off-the-shelf \smt solver for quantifier-free linear real arithmetic
\cite{DBLP:conf/fmcad/GanaiI09}, or focus on particular kinds of
constraints like convex constraints \cite{DBLP:conf/fmcad/NuzzoPSS10}.
In the implementation of many of these approaches computations are
performed with floating-point arithmetic. In order to address the
ever-present concern that numerical errors can result in incorrect
answers, the framework of $\delta$-complete decision procedures has
been proposed \cite{DBLP:conf/cade/GaoAC12,DBLP:conf/cade/GaoKC13}. In
another line of research, as opposed to numerically-driven approaches,
symbolic techniques from algebraic geometry such as the aforementioned
CAD \cite{JovanovicMoura2012IJCAR}, Gr\"obner bases
\cite{DBLP:conf/cai/JungesLCA13,DBLP:conf/dagstuhl/PassmoreMJ10},
Handelman's representations \cite{DBLP:conf/vmcai/MarechalFKMP16} or
virtual substitution \cite{DBLP:conf/fct/CorziliusA11} have been
successfully adapted to \sat and \smt. As a result, several libraries
and toolboxes have been made publicly available for the development of
symbolically-driven solvers
\cite{DBLP:conf/sat/CorziliusLJA12,DBLP:conf/cade/MouraP13,DBLP:conf/sat/CorziliusKJSA15}.

On the other hand, when variables have to take integer values, even
the problem of solving a single polynomial equation is undecidable
(Hilbert's 10th problem, \cite{Cooper2004}). Despite this theoretical
limitation, and following a similar direction to the real case,
several incomplete methods that exploit the progress in \sat and \smt
have been proposed for dealing with integer polynomial constraints.
The common idea of these approaches is to reduce instances of this
kind of formulas into problems of a simpler language that can be
straightforwardly handled by existing \sat/\smt systems, e.g.,
propositional logic \cite{DBLP:conf/sat/FuhsGMSTZ07}, linear
bit-vector arithmetic \cite{DBLP:conf/lpar/ZanklM10} or linear integer
arithmetic \cite{Borrallerasetal2011JAR}. All these techniques are
oriented towards satisfiability, which makes them convenient in
applications where finding solutions is more relevant than proving
that none exists (e.g., in verification when generating ranking
functions \cite{DBLP:conf/fmcad/LarrazORR13}, invariants
\cite{DBLP:conf/vmcai/LarrazRR13} or other inductive properties
\cite{DBLP:conf/cav/LarrazNORR14,DBLP:conf/fmcad/BrockschmidtLOR15}).

In this article\footnote{This is the extended version of the
  conference paper presented at SAT\,'14
  \cite{DBLP:conf/sat/LarrazORR14}.} we build upon our previous method
\cite{Borrallerasetal2011JAR} for deciding \SMT(QF-NIA) by reduction
to \SMT(QF-LIA). In that work, non-linear monomials are linearized by
abstracting them with fresh variables and by performing case splitting
on integer variables with finite domain. In the case in which
variables do not have finite domains, artificial ones are introduced
by imposing a lower and an upper bound. While the underlying
\SMT(QF-LIA) solver cannot find a solution (and the time limit has not
been exceeded yet), \emph{domain relaxation} is applied: some domains
are made larger by weakening the bounds. To guide which bounds have to
be relaxed from one iteration to the following one, unsatisfiable
cores are employed: at least one of the artificial bounds that appear
in the unsatisfiable core should be weaker. Unfortunately, although
unsatisfiable cores indicate which bounds should be relaxed, they
provide no hint on how large the new domains have to be made. This is
of paramount importance, since the size of the new linearized formula (and
therefore the time needed to determine its satisfiability) can
increase significantly depending on the number of new cases that must
be added.

A way to circumvent this difficulty could be to find alternative
techniques to the unsatisfiable cores which, when a solution with the
current domains cannot be found, provide more complete information for
the domain relaxation. In this paper we propose such alternative
techniques. The key idea is that an assignment of numerical values to
variables that is ``closest'' to being a true solution (according to
some metric) can be used as a reference as regards to how one should
enlarge the domains. Thus, the models generated by the \SMT(QF-LIA)
engine are put in use in the search of solutions of the original
non-linear problem, with a similar spirit to
\cite{DBLP:journals/entcs/MouraB08} for combining theories or to the
model-constructing satisfiability calculus of
\cite{MouraJovanovic2013VMCAI}.

However, as pointed out above, in our case we are particularly
interested in \emph{minimal models}, namely those that minimize a cost
function that measures how far assignments are from being a true
solution to the non-linear problem. Minimal models have long been
studied in the case of propositional logic
\cite{DBLP:journals/amai/Ben-EliyahuD96,DBLP:journals/ai/Ben-Eliyahu-Zohary05,DBLP:conf/ecai/SohI10}.
In \smt, significant advancements have been achieved towards solving
the optimization problems of \emph{Maximum Satisfiability Modulo
  Theories} (\maxsmt,
\cite{NieuwenhuisOliveras2006SAT,CimattiFranzenGriggioSebastianiStenico2010TACAS})
and \emph{Optimization Modulo Theories} (\omt,
\cite{Vendrell2012MasterThesis,DBLP:journals/tocl/SebastianiT15b}).
Thanks to this research, several \smt systems are currently offering
optimization functionalities
(\cite{DBLP:conf/cav/SebastianiT15a,DBLP:conf/popl/LiAKGC14,DBLP:conf/tacas/BjornerPF15}).
Taking advantage of these recent developments, in this paper we make
the following contributions:

\begin{longenum}

\item In the context of solving \SMT(QF-NIA), we present different
  heuristics for guiding the domain relaxation step by means of the
  analysis of minimal models. More specifically, we consider two
  different cost functions:

  \begin{itemize}

  \item the number of violated artificial domain bounds (leading to
    \maxsmt problems);

  \item the distance with respect to the artificial domains (leading
    to \omt problems).

  \end{itemize}

  We evaluate these model-guided heuristics experimentally with an
  exhaustive benchmark set and compare them with other techniques for
  solving \SMT(QF-NIA). The results of this evaluation show the
  potential of the method.

\item Based on the results of the aforementioned experiments,
  we extend our best approach for \SMT(QF-NIA) to handle problems in
  \MAXSMT(QF-NIA).

\item Finally we apply our \MAXSMT(QF-NIA) techniques to solve \smt
  and \maxsmt problems in the following fragment of quantified
  non-linear arithmetic: $\exists \forall$ formulas where $\exists$
  variables are of integer type and $\forall$ variables are of real
  type, and non-linear monomials cannot contain the product of two
  real variables. Formulas of this kind appear commonly in
  verification and synthesis applications \cite{Dutertre:smt2015}, for
  example in control and priority synthesis
  \cite{ChengShankarRuessBensalem2013}, reverse engineering of
  hardware \cite{DBLP:conf/fmcad/GasconSDTJM14} and program synthesis
  \cite{DBLP:conf/cade/TiwariGD15}.

\end{longenum}

This paper is structured as follows. Section \ref{sec:preliminaries}
reviews basic background on \smt, \maxsmt and \omt, and also on our
previous approach in \cite{Borrallerasetal2011JAR}. In Section
\ref{sec:domain_relaxation} two different heuristics for guiding the
domain relaxation step are presented, together with experiments and
several possible variants. Then Section \ref{sec:maxsmt_nia}
proposes an extension of our techniques from \SMT(QF-NIA) to
\MAXSMT(QF-NIA). In turn, in Section \ref{sec:maxsmt_ea} our
\MAXSMT(QF-NIA) approach is applied to solving \maxsmt problems with
$\exists \forall$ formulas. Finally, Section \ref{sec:conclusions}
summarizes the conclusions of this work and sketches lines for future
research.

\section{Preliminaries}
\label{sec:preliminaries}

\subsection{\smt, \maxsmt and \omt}
\label{sec:smt_maxsmt_omt}

Let $\cP$ be a fixed finite set of \emph{propositional variables}.  If
$p \in \cP$, then $p$ and $\lnot p$ are \emph{literals}.  The
\emph{negation} of a literal $l$, written $\lnot l$, denotes $\lnot p$
if $l$ is $p$, and $p$ if $l$ is $\lnot p$.  A \emph{clause} is a
disjunction of literals $l_1 \lor\cdots\lor l_n$. A (CNF)
\emph{propositional formula} is a conjunction of clauses
$C_1 \land\cdots\land C_n$. The problem of \emph{propositional
  satisfiability} (abbreviated \sat) consists in, given a
propositional formula, to determine whether it is
\emph{satisfiable}, i.e., if it has a \emph{model}: an assignment of
Boolean values to variables that satisfies the formula.

The \emph{satisfiability modulo theories (SMT)} problem is a
generalization of \sat. In \smt, one has to decide the satisfiability
of a given (usually, quantifier-free) first-order formula with respect
to a background theory. In this setting, a model (which we may also
refer to as a \emph{solution}) is an assignment of values from the
theory to variables that satisfies the formula.
Examples of theories are \emph{quantifier-free linear
  integer arithmetic (QF-LIA)}, where literals are linear inequalities
over integer variables, and the more general \emph{quantifier-free
  non-linear integer arithmetic (QF-NIA)}, where literals are
polynomial inequalities over integer variables. Unless otherwise
stated, in this paper we will assume that variables are all of integer
type.

Another generalization of \sat is \emph{\MAXSAT}
\cite{DBLP:series/faia/LiM09}, which extends the problem by asking for
more information when the formula turns out to be unsatisfiable:
namely, the \maxsat problem consists in, given a formula $F$, to find
an assignment such that the number of satisfied clauses in $F$ is
maximized, or equivalently, that the number of falsified clauses is
minimized. This problem can in turn be generalized in a number of
ways. For example, in \emph{weighted \maxsat} each clause of $F$ has a
\emph{weight} (a positive natural or real number), and then the goal
is to find the assignment such that the \emph{cost}, i.e., the sum of
the weights of the falsified clauses, is minimized. Yet a further
extension of \maxsat is the \emph{partial weighted \maxsat} problem,
where clauses in $F$ are either weighted clauses as explained above,
called \emph{soft clauses} in this setting, or clauses without
weights, called \emph{hard clauses}. In this case, the problem
consists in finding the model of the hard clauses such that the sum of
the weights of the falsified soft clauses is minimized. Equivalently,
hard clauses can also be seen as soft clauses with infinite weight.

The problem of \emph{\MAXSMT} merges \maxsat and \smt, and is defined
from \smt analogously to how \maxsat is derived from \sat. Namely, the
\emph{Max-SMT} problem consists in, given a set of pairs
$\{ [C_1,\weight_1], \ldots,[C_m,\weight_m] \}$, where each $C_i$ is a
clause and $\weight_i$ is its weight (a positive number or infinity),
to find an assignment that minimizes the sum of the weights of the
falsified clauses in the background theory. As in \smt, in this
context we are interested in assignments of values from the theory to
variables.

Finally, the problem of \emph{Optimization Modulo Theories (OMT)} is
similar to \maxsmt in that they are both optimization problems, rather
than decision problems. It consists in, given a formula $F$ involving
a particular numerical variable called $\mathit{cost}$, to find the
model of $F$ such that the value assigned to $\mathit{cost}$ is
minimized. Note that this framework allows one to express a wide
variety of optimization problems (maximization, piecewise linear
functions, etc.).

\subsection{Solving \SMT(QF-NIA) with Unsatisfiable Cores}
\label{sec:unsatisfiable_cores}

In \cite{Borrallerasetal2011JAR}, we proposed a method for solving
\SMT(QF-NIA) problems based on encoding them into \SMT(QF-LIA). The basic
idea is to linearize each non-linear monomial in the formula by
applying a case analysis on the possible values of some of its
variables. For example, if the monomial $x^2 yz$ appears in the input
QF-NIA formula and $x$ must satisfy $0\leq x \leq 2$, we can
introduce a fresh variable $v_{x^2yz}$, replace the occurrences of
$x^2 yz$ by $v_{x^2 yz}$ and add to the clause set the following three
\emph{case splitting clauses}:

%
\begin{center}
  $
  \begin{array}{llll}
    x=0 & \,\limplies\, & v_{x^2yz}=0 \, ,\\
    x=1 & \,\limplies\, & v_{x^2yz}=yz \, ,\\
    x=2 & \,\limplies\, & v_{x^2yz}=4yz \, .\\
  \end{array}
  $
\end{center}

In turn, new non-linear monomials may appear, e.g., $yz$ in this
example. All non-linear monomials are handled in the same way until a
formula in QF-LIA is obtained, for which efficient decision procedures
exist
\cite{DutertreMoura2006CAV,DBLP:journals/fmsd/DilligDA11,DBLP:journals/jsat/Griggio12}.

Note that, in order to linearize a non-linear monomial, there must be
at least one variable in it which is both lower and upper bounded.
When this property does not hold, new \emph{artificial} domains can be
introduced for the variables that require them (for example, for
unbounded variables one may take $\{-1, 0, 1\}$). In principle, this
implies that the procedure is no longer complete, since a linearized
formula with artificial bounds may be unsatisfiable while the original
QF-NIA formula is actually satisfiable. A way to overcome this problem
is to proceed iteratively: variables start with bounds that make the
size of their domains small, and then the domains are enlarged on
demand if necessary, i.e., if the formula turns out to be
unsatisfiable. The decision of which bounds are to be relaxed is
heuristically taken based on the analysis of an \emph{unsatisfiable
  core} (an unsatisfiable subset of the clause set) that is obtained
when the solver reports unsatisfiability. There exist many techniques
in the literature for computing unsatisfiable cores (see, e.g.,
\cite{AchaEtAlAICOM10} for a sample of them). In
\cite{Borrallerasetal2011JAR} we employed the well-known simple and
effective approach of \cite{ZhangMalikDATE2003}, consisting in writing
a trace on disk and extracting a resolution refutation, whose leaves
form an unsatisfiable core. Note that the method tells \emph{which}
bounds should be weakened, but does not provide any guidance in regard
to \emph{how large} the change on the bounds should be. This is
critical, as the size of the formula in the next iteration (and so the
time required to determine its satisfiability) can grow significantly
depending on the number of new case splitting clauses that have to be
added. Therefore, in lack of a better strategy, a typical heuristic is
to decrement or increment the bound (for lower bounds and for upper
bounds, respectively) by a constant value.

\begin{algorithm}[t]
\begin{lstlisting}
status $\solveSMTQFNIACores$(formula $F_0$) { // returns whether $F_0$ is satisfiable
  $\bounds$ = $\artificialBounds$($F_0$);                                  // $\bounds$ are artificial bounds enough to linearize $F_0$
  $F$ = $\linearize$($F_0$, $\bounds$);
  while (not timed_out()) {
    $\langle \status, \core \rangle$ = $\solveSMTQFLIA$($F$, $\bounds$);         // unsatisfiable core $\mathit{UC}$ computed here for simplicity
    if ($\status$ == $\Sat$) return $\Sat$;                        // is $F \land \bounds$ satisfiable?
    else if ($\bounds \cap \core$ == $\emptyset$) return $\Unsat$;
    else {
      $\bounds'$ = $\relaxDomainsCores$($\bounds$, $\core$);               // at least one in the intersection is relaxed
      $F$ = $\update$($F$, $\bounds$, $\bounds'$);                                 // add case splitting clauses
      $\bounds$ = $\bounds'$;
    }
  }
  return $\Unknown$;
}
\end{lstlisting}
\caption{Algorithm for solving \SMT(QF-NIA) with unsatisfiable cores}
\label{alg:solve_SMT_QF_NIA_cores}
\end{algorithm}

\begin{algorithm}[t]
\begin{lstlisting}
set $\artificialBounds$(formula $F_0$) {          // returns the artificial bounds for linearization
  $S$ = choose_linearization_variables($F_0$);  // choose enough variables to linearize $F_0$
  $B$ = $\emptyset$;                                    // set of artificial bounds
  for ($V$ in $S$) {
    if (lower_bound($V$, $F_0$) $\,$== $\bot$)             // cannot find lower bound of $V$ in $F_0$
       $B$ = $B$ $\cup$ { $V \geq L$};                          // for a parameter $L$, e.g. $L = -1$
    if (upper_bound($V$, $F_0$) == $\bot$)             //  cannot find upper bound of $V$ in $F_0$
       $B$ = $B$ $\cup$ { $V \leq U$};                          // for a parameter $U$, e.g. $U = 1$
  }
  return $B$;
}
\end{lstlisting}
\caption{Procedure \artificialBounds}
\label{alg:artificial_bounds}
\end{algorithm}

\begin{algorithm}[t]
\begin{lstlisting}
formula $\linearize$(formula $F_0$, set $\bounds$) {              // returns the linearization of $F_0$
  $\nonlinear$ = nonlinear_monomials($F_0$);
  $F = F_0$;
  while ($\nonlinear \neq \emptyset$) {
    let $Q$ in $\nonlinear$;                                     // non-linear monomial to be linearized next
    $V_Q$ = fresh_variable();
    $F$ = replace($Q$, $F$, $V_Q$);                             // replace all occurrences of $Q$ in $F$ by $V_Q$
    $C$ = $\emptyset$;                                          // clauses of the case splitting
    $V$ = linearization_variable($Q$);                  // choose a finite domain variable in $Q$ to linearize
    for ($K$ in [lower_bound($V$, $F_0 \cup \bounds$), upper_bound($V$, $F_0 \cup \bounds$)])
      $C$ = $C$ $\cup$ { $V = K \limplies V_Q =$ evaluate($Q$, $V$, $K$)};
    $F$ = $F$ $\cup$ $C$;
    $\nonlinear$ = $\nonlinear$ - $\{Q\}$ $\cup$ nonlinear_monomials($C$);      // new non-linear monomials may be introduced
  }
  return $F$;
}
\end{lstlisting}
\caption{Procedure \linearize}
\label{alg:linearize}
\end{algorithm}

\begin{algorithm}[t]
\begin{lstlisting}
set $\relaxDomainsCores$(set $\bounds$, set $\core$) {                        // returns the new set of artificial bounds
  let $S \subseteq \bounds \cap \core$ such that $S \neq \emptyset$;
  $\bounds' = \bounds$;
  for ($V\geq L$ in $S$) $\hspace{0.1mm}$ $\bounds'$ $\!$= $\bounds'$ - {$V \geq L$} $\cup$ {$V \geq L'$};     // e.g. $L' = L - K_L$ $\hspace{0.1mm}$ for a parameter $K_L > 0$
  for ($V\leq U$ in $S$) $\bounds'$ = $\bounds'$ - {$V \leq U$} $\cup$ {$V \leq U'$}; $\!$    // e.g. $U' = U + K_U$ for a parameter $K_U > 0$
  return $\bounds'$;
}
\end{lstlisting}
\caption{Procedure \relaxDomainsCores}
\label{alg:relax_domains_cores}
\end{algorithm}

\begin{algorithm}[t]
\begin{lstlisting}
formula $\update$(formula $F$, set $\bounds$, set $\bounds'$) { // adds cases when relaxing the bounds from $\bounds$ to $\bounds'$
  $F'$ = $F$;
  for ($V$ $\;$such that$\;$ $V\geq L$ in $\bounds$ $\;$and$\;$ $V\geq L'$ in $\bounds'$ $\;$and$\;$ $L \neq L'$)
    for ($K$ in [$L'$, $L-1$])
       for ($Q$ such that $V$ == linearization_variable($Q$))       // $V$ was used to linearize monomial $Q$
          $F'$ = $F'$ $\cup$ { $V = K \limplies V_Q =$ evaluate($Q$, $V$, $K$)};     // $V_Q$ is the variable standing for $Q$

  for ($V$ $\;$such that$\;$ $V\leq U$ in $\bounds$ $\;$and$\;$ $V\leq U'$ in $\bounds'$ $\;$and$\;$ $U \neq U'$)
    for ($K$ in [$U+1$, $U'$])
       for ($Q$ such that $V$ == linearization_variable($Q$))
          $F'$ = $F'$ $\cup$ { $V = K \limplies V_Q =$ evaluate($Q$, $V$, $K$)};

  return $F'$;
}
\end{lstlisting}
  \caption{Procedure \update}
  \label{alg:update}
\end{algorithm}

Procedure \solveSMTQFNIACores in Algorithm
\ref{alg:solve_SMT_QF_NIA_cores} describes more formally the overall
algorithm from \cite{Borrallerasetal2011JAR} for solving \SMT(QF-NIA).
First, the required artificial bounds are computed (procedure
\artificialBounds, with pseudo-code in Algorithm
\ref{alg:artificial_bounds}). Then the linearized formula (procedure
\linearize, with pseudo-code in Algorithm \ref{alg:linearize})
together with the artificial bounds are passed to an \SMT(QF-LIA)
solver (procedure \solveSMTQFLIA), which tests if their conjunction is
satisfiable \footnote{Note that, in this formulation, the
  linearization consists of the clauses of the original formula after
  replacing non-linear monomials by fresh variables, together with the
  case splitting clauses. On the other hand, it does \emph{not}
  include the artificial bounds, which for the sake of presentation
  are kept as independent objects. }. If the solver returns $\Sat$, we
are done. If the solver returns $\Unsat$, then an unsatisfiable core
is also computed. If this core does not contain any of the artificial
bounds, then the original non-linear formula must be unsatisfiable,
and we are done too. Otherwise, at least one of the artificial bounds
appearing in the core must be chosen for relaxation (procedure
\relaxDomainsCores, with pseudo-code in Algorithm
\ref{alg:relax_domains_cores}). Once the domains are enlarged and the
appropriate case splitting clauses are added (procedure \update, with
pseudo-code in Algorithm \ref{alg:update}), the new linearized formula
is tested for satisfiability again, and the process is repeated
(typically, while a predetermined time limit is not exceeded). We
refer the reader to \cite{Borrallerasetal2011JAR} for further details.

\begin{example}
\label{ex:running}
Let $F_0$ be the formula
$$tx + y \,\geq\, 4 \,\;\land\;\, t^2 + x^2 + w^2 + y^2 \,\leq\, 12,$$
where variables $t, x, w, y$ are integer. Let us also assume that we
introduce the following artificial bounds so as to linearize:
$\bounds \equiv -1 \leq t, x, w, y \leq 1$. Now a linearization $F$ of
$F_0$ could be for example:

\begin{center}

$
\begin{array}{ll}
v_{tx} + y \geq 4 \;\land\; v_{t^2} + v_{x^2} + v_{w^2} + v_{y^2} \leq 12 \;\;\land\medskip\\
\begin{array}{lcl c lcl c}
( t =-1 & \limplies & v_{tx} = -x   ) & \;\;\;\land\;\;\; &  \\
( t = 0 & \limplies & v_{tx} =  0   ) & \;\;\;\land\;\;\; &  \\
( t = 1 & \limplies & v_{tx} =  x   ) & \;\;\;\land\;\;\; &  \\
\\
( t =-1 & \limplies & v_{t^2} =  1  ) & \;\;\;\land\;\;\; &  ( w =-1 & \limplies & v_{w^2} =  1  ) & \;\;\land\\
( t = 0 & \limplies & v_{t^2} =  0  ) & \;\;\;\land\;\;\; &  ( w = 0 & \limplies & v_{w^2} =  0  ) & \;\;\land\\
( t = 1 & \limplies & v_{t^2} =  1  ) & \;\;\;\land\;\;\; &  ( w = 1 & \limplies & v_{w^2} =  1  ) & \;\;\land\\
\\
( x =-1 & \limplies & v_{x^2} =  1  ) & \;\;\;\land\;\;\; &  ( y =-1 & \limplies & v_{y^2} =  1  ) & \;\;\land\\
( x = 0 & \limplies & v_{x^2} =  0  ) & \;\;\;\land\;\;\; &  ( y = 0 & \limplies & v_{y^2} =  0  ) & \;\;\land\\
( x = 1 & \limplies & v_{x^2} =  1  ) & \;\;\;\land\;\;\; &  ( y = 1 & \limplies & v_{y^2} =  1  ) & \;\;\land\\
\end{array}
\end{array}
$
\end{center}
where $v_{tx}, v_{t^2}, v_{x^2}, v_{w^2}, v_{y^2}$ are fresh
integer variables standing for the non-linear monomials in the
respective subscripts.

In this case the formula $F \land \bounds$ turns out to be
unsatisfiable. For instance, the \SMT(QF-LIA) solver could produce the
following unsatisfiable core:

\bigskip

\begin{math}
\{  v_{tx} + y \geq 4,
\end{math}

\smallskip

\begin{math}
\begin{array}{lclcrclcrcl}
t =-1 & \limplies & v_{tx} = -x, & \quad &  y & \leq & 1,                        \smallskip \\
t = 0 & \limplies & v_{tx} =  0, & \quad & -1 & \leq & t, & &  -1  & \leq & x,   \smallskip \\
t = 1 & \limplies & v_{tx} =  x, & \quad &  t & \leq & 1, & &   x  & \leq & 1 \} \smallskip \\
\end{array}
\end{math}

\smallskip


\bigskip

Intuitively, if $|t|, |x|, y \leq 1$, then it cannot be the
case that $tx + y \geq 4$. At this stage, one has to relax at least
one of the artificial bounds in the core, for example $x \leq 1$.
Notice that, on the other hand, the core does not provide any help in regard to
deciding the new upper bound for $x$. If, e.g., we chose that it were
$x \leq 4$, then $x \leq 4$ would replace $x \leq 1$ in the set of
artificial bounds $\bounds$, and the following clauses would be added
to the linearization $F$:

\begin{center}
$
\begin{array}{lcl}
x = 2 & \limplies & v_{x^2} = 4\\
x = 3 & \limplies & v_{x^2} = 9\\
x = 4 & \limplies & v_{x^2} = 16\\
\end{array}
$
\end{center}

In the next iteration one could already find solutions to the
non-linear formula $F_0$, for instance,
$t = v_{t^2} = w = v_{w^2} = y = v_{y^2} = 1$, $x = v_{tx} = 3$, and
$v_{x^2} = 9$.\hfill$\blacksquare$
\end{example}

\section{Solving \SMT(QF-NIA) with Minimal Models}
\label{sec:domain_relaxation}

Taking into account the limitations of the method based on cores when
domains have to be extended, in this section we present a model-guided
approach to perform this step. Namely, we propose to replace the
satisfiability check in linear arithmetic with an optimization call,
so that the best model found by the linear solver can be used as a
reference for relaxing the bounds.

This is the key idea of the procedure
\solveSMTQFNIAMinModels for solving \SMT(QF-NIA)
shown in Algorithm \ref{alg:solve_SMT_QF_NIA_min_models} (cf.
Algorithm \ref{alg:solve_SMT_QF_NIA_cores}; note that all
subprocedures except for \optimizeQFLIA and
\relaxDomainsMinModels are the same). Now the
\SMT(QF-LIA) black box (procedure \optimizeQFLIA) does
not just decide satisfiability, but finds the minimal model of its
input formula $F$ according to a certain cost function. If this model
does not satisfy the original non-linear formula, then it can be
employed as a hint in the domain relaxation (procedure
\relaxDomainsMinModels, with pseudo-code in Algorithm
\ref{alg:relax_domains_min_models}) as follows. Since the non-linear
formula is not satisfied, it must be the case that some of the
artificial bounds are not respected by the minimal model. By gathering
these bounds, a set of candidates to be relaxed is obtained, as in the
approach of Section \ref{sec:unsatisfiable_cores}. However, and most
importantly, unlike with unsatisfiable cores now for each of these
bounds a new value can be guessed too: one just needs to take the
corresponding variable and enlarge its domain up to the value assigned
in the minimal model. That is, let us assume that $V \leq U$ is an
artificial upper bound on variable $V$ that is falsified in the
minimal model. If $V$ is assigned value $U'$ in that model (and,
hence, $U < U'$), then $V \leq U'$ becomes the new upper bound of $V$.
A similar construction applies for lower bounds.

The intuition behind this approach is that the cost function should
measure how far assignments are from being a solution to the original
non-linear formula. Formally, the function must be non-negative and
have the property that the models of the linearized formula with cost
0 are true models of the original non-linear formula:

\begin{theorem}
  \label{th:correctness}
  Let $F_0$ be an arbitrary formula in QF-NIA, and $F$ be any
  linearization of $F_0$ in QF-LIA obtained using the procedure
  \linearize with artificial bounds $B$.

  A function $\cost$ that takes as input the models of $F$ is
  \emph{admissible} if:

  \begin{enumerate}
  \item $\cost(M) \geq 0$ for any $M$ model of $F$;
  \item If $\cost(M) = 0$ then $M$ is a model of $F_0$.
  \end{enumerate}

  If the cost functions in procedure \solveSMTQFNIAMinModels
  are admissible then the procedure is correct. That is, given a
  formula $F_0$ in QF-NIA:

  \begin{enumerate}

  \item if \solveSMTQFNIAMinModels($F_0$) returns $\Sat$ then
    formula $F_0$ is satisfiable; and

  \item if \solveSMTQFNIAMinModels($F_0$) returns $\Unsat$ then
    formula $F_0$ is unsatisfiable.

  \end{enumerate}

\end{theorem}

\begin{proof}
  \begin{longenum}

  \item Let us assume that \solveSMTQFNIAMinModels($F_0$) returns
    $\Sat$. Then there is a set of artificial bounds $B$ such that
    $F$, the linearization of $F_0$ using $B$, satisfies the
    following: \optimizeQFLIA($F$, $B$) returns a model $\model$ of
    $F$ such that $\cost(\model) = 0$. As $\cost$ is admissible we
    have that $\model$ is a model of $F_0$.

  \item Let us assume that
    \solveSMTQFNIAMinModels($F_0$) returns
    $\Unsat$. Then there is a set of artificial bounds $B$ such that
    $F$, the linearization of $F_0$ using $B$, satisfies that
    \optimizeQFLIA($F$, $B$) returns $\Unsat$. By the
    specification of \optimizeQFLIA, this means that $F$
    is unsatisfiable. But since only case splitting clauses are added
    in the linearization, any model of $F_0$ can be extended to a
    model of $F$. By reversing the implication we conclude that $F_0$
    must be unsatisfiable.
  \end{longenum}
\end{proof}

\begin{algorithm}[t]
\begin{lstlisting}
status $\solveSMTQFNIAMinModels$(formula $F_0$){ // returns whether $F_0$ is satisfiable
  $\bounds$ = $\artificialBounds$($F_0$);                                            // $\bounds$ are artificial bounds enough to linearize $F_0$
  $F$ = $\linearize$($F_0$, $\bounds$);
  while (not timed_out()) {

    // If $\status == \Unsat$ then $F$ is $\Unsat$
    // If $\status == \Sat$ then $\model$ is a model of $F$ minimizing function $\cost$ below among all models of $F$
    $\langle \status, \model \rangle$ = $\optimizeQFLIA$($F$, $\bounds$);

    if ($\status$ == $\Unsat$)       return $\Unsat$;
    else if ($\cost(\model)$ == $0$) return $\Sat$;
    else {
      $\bounds'$ = $\relaxDomainsMinModels$($\bounds$, $\model$);
      $F$ = $\update$($F$, $\bounds$, $\bounds'$);                                        // add case splitting clauses
      $\bounds$ = $\bounds'$;
  } }
  return $\Unknown$;
}
\end{lstlisting}
  \caption{Algorithm for solving \SMT(QF-NIA) with minimal models}
  \label{alg:solve_SMT_QF_NIA_min_models}
\end{algorithm}

\begin{algorithm}[t]
\begin{lstlisting}
set $\relaxDomainsMinModels$(set $\bounds$, map $\model$) {   // returns the new set of artificial bounds
  let $S \subseteq \{b \;|\; b \in \bounds, \model \not\models b\}$ such that $S \neq \emptyset$;    // choose among bounds violated by the model
  $\bounds' = \bounds$;
  for ($V\geq L$ in $S$) $\hspace{0.1mm}$ $\bounds'$ $\!$= $\bounds'$ - {$V \geq L$} $\hspace{0mm}$ $\!\cup$ {$V \geq \model(V)$}; // $L > \model(V)$ as $\model \not\models$ $V\geq L$
  for ($V\leq U$ in $S$) $\bounds'$ = $\bounds'$ - {$V \leq U$} $\cup$ {$V \leq \model(V)$}; $\!$ // $U < \model(V)$ as $\model \not\models$ $V\leq U$
  return $\bounds'$;
}
\end{lstlisting}
\caption{Procedure \relaxDomainsMinModels}
\label{alg:relax_domains_min_models}
\end{algorithm}

Under the assumption that cost functions are admissible, note that, if
at some iteration in procedure
\solveSMTQFNIAMinModels there are models of the
linearization with null cost (hence satisfying the original non-linear
formula), then the search is over: \optimizeQFLIA will
return such a model, as it minimizes a non-negative cost function.

In what follows we propose two different admissible (classes of) cost
functions: the \emph{number} of violated artificial bounds (Section
\ref{sec:maxsmt_la}), and the \emph{distance} with respect to the
artificial domains (Section \ref{sec:omt_la}). In both cases, to
complete the implementation of \solveSMTQFNIAMinModels the only
procedure that needs to be defined is \optimizeQFLIA, since procedure
\relaxDomainsMinModels is independent of the cost function (see
Algorithm \ref{alg:relax_domains_min_models}).

\subsection{A \MAXSMT(QF-LIA) Approach to Domain Relaxation}
\label{sec:maxsmt_la}

As sketched out above, a possibility is to define the cost of an
assignment as the number of violated artificial bounds. A natural way
of implementing this is to transform the original non-linear formula
into a linearized weighted formula and use a \MAXSMT(QF-LIA) tool. In
this setting, the clauses of the linearization are hard, while the
artificial bounds are considered to be soft (e.g., with weight $1$ if
we literally count the number of violated bounds). Procedure
\optimizeQFLIAMaxSMT is described formally in Algorithm
\ref{alg:optimize_QF_LIA:maxsmt}. It is worth highlighting that not
only is the underlying \MAXSMT(QF-LIA) solver required to report the
optimum value of the cost function, but it must also produce an
assignment in the theory for which this optimum value is attained (so
that it can be used in the domain relaxation). A direct and effective
way of accomplishing this task is by performing branch-and-bound on
top of an \SMT(QF-LIA) solver, as done in
\cite{NieuwenhuisOliveras2006SAT}\footnote{Other approaches could also
  employed for solving \MAXSMT(QF-LIA); for example, one could
  iteratively obtain unsatisfiable cores and add relaxation variables
  and cardinality or pseudo-Boolean constraints to the instance until
  a $\Sat$ answer is obtained
  \cite{FuMalik2006SAT,DBLP:journals/ai/AnsoteguiBL13,DBLP:journals/constraints/MorgadoHLPM13}.
  Nevertheless, here we opted for branch-and-bound for its simplicity
  and because it can be easily adapted to meet the requirements for
  solving \MAXSMT(QF-NIA); see Section \ref{sec:maxsmt_nia}.}.

\begin{algorithm}[t]
\begin{lstlisting}
$\langle$ status, map $\rangle$ $\optimizeQFLIAMaxSMT$(formula $F$, set $\bounds$) {
  $F' = F$;
  for ($V \geq L$ in $\bounds$)
    $F'$ = $F'$ $\cup$ $\{ [V \geq L, 1]\}$;                        // added as a soft clause, e.g. with weight 1
  for ($V \leq U$ in $\bounds$)
    $F'$ = $F'$ $\cup$ $\{ [V \leq U, 1]\}$;                        // added as a soft clause, e.g. with weight 1

  return $\solveMaxSMTQFLIA$($F'$);     // call to \maxsmt solver
}
\end{lstlisting}
\caption{Procedure \optimizeQFLIAMaxSMT based on \MAXSMT(QF-LIA)}
\label{alg:optimize_QF_LIA:maxsmt}
\end{algorithm}

The next lemma justifies, together with Theorem \ref{th:correctness},
that procedure \solveSMTQFNIAMinModels, when
instantiated with \optimizeQFLIAMaxSMT, is correct:

\begin{lemma}
\label{lemma:optimize_QF_LIA:maxsmt}
  Let $F_0$ be an arbitrary formula in QF-NIA, and $F$ be any
  linearization of $F_0$ in QF-LIA obtained using the procedure
  \linearize with artificial bounds $B$.

  The function $\cost$ that takes as an input a model $\model$ of $F$
  and returns the number of bounds from $\bounds$ that are not
  satisfied by $\model$ is admissible.
\end{lemma}

\begin{proof}
  It is obvious that the function is non-negative. It remains to be
  proved that, if $\cost(\model) = 0$, then $\model \models F_0$.
  Indeed, if $\cost(\model) = 0$ then all bounds in $\bounds$ are
  satisfied. But since $F$ is a linearization of $F_0$ with artificial
  bounds $B$, we have that all additional variables standing for
  non-linear monomials have values in $\model$ that are consistent
  with the theory. Hence, we conclude that $\model \models F_0$.
\end{proof}

Regarding the weights of the soft clauses, as can be observed from the
proof of Lemma \ref{lemma:optimize_QF_LIA:maxsmt}, it is not necessary
to have unit weights. One may use different values, provided they are
positive, and then the cost function corresponds to a weighted sum.
Moreover, note that weights can be different from one iteration of the
loop of \solveSMTQFNIAMinModels to the next one.

\begin{example}
  \label{ex:maxsmt}
  Let us consider the same formula as in Example \ref{ex:running}:
  $$tx + y \,\geq\, 4 \;\land\; t^2 + x^2 + w^2 + y^2 \,\leq\, 12.$$
  Recall that, in this case, the artificial bounds are
  $-1 \leq t, x, w, y \leq 1$. We obtain the weighted formula
  consisting of the clauses of $F$ (as defined in Example
  \ref{ex:running}) as hard clauses, and
\begin{center}
$
\begin{array}{rcrcrcrcrcrcrc} \\
\left[-1 \leq t, 1\right]  & \land &  \left[-1 \leq x, 1\right]  & \land &  \left[-1 \leq w, 1\right]  & \land &  \left[-1 \leq y, 1\right] & \;\land \smallskip\\
\left[ t \leq 1, 1\right]  & \land &  \left[ x \leq 1, 1\right]  & \land &  \left[ w \leq 1, 1\right]  & \land &  \left[ y \leq 1, 1\right]\,,
\end{array}
$
\end{center}
as soft clauses (written following the format
$[\mathrm{clause}, \mathrm{weight}]$).

In this case minimal solutions have cost $1$: at least one of the
artificial bounds has to be violated so as to satisfy
$v_{tx} + y \geq 4$. For instance, the \MAXSMT(QF-LIA) solver could
return the assignment: $t = v_{t^2} = 1$, $x = v_{tx} = 4$ and
$w = v_{w^2} = y = v_{y^2} = v_{x^2} = 0$, where the only soft
clause that is violated is $[ x \leq 1, 1]$. Note that, as $x = 4$ is
not covered by the case splitting clauses for $v_{x^2}$, the values of
$v_{x^2}$ and $x$ are unrelated. Now the new upper bound for $x$ would
become $x \leq 4$ (so the soft clause $[ x \leq 1, 1]$ would be
replaced by $[ x \leq 4, 1]$), and similarly to Example
\ref{ex:running}, the following hard clauses would be added:

\begin{center}
$
\begin{array}{lcl}
x = 2 & \limplies & v_{x^2} = 4\\
x = 3 & \limplies & v_{x^2} = 9\\
x = 4 & \limplies & v_{x^2} = 16\\
\end{array}
$
\end{center}

As seen in Example \ref{ex:running}, in the next iteration there are
solutions with cost 0, e.g.,
$t = v_{t^2} = w = v_{w^2} = y = v_{y^2} = 1$, $x = v_{tx} = 3$ and
$v_{x^2} = 9$.\hfill$\blacksquare$
\end{example}

One of the disadvantages of this approach is that potentially the
\MAXSMT(QF-LIA) solver could return models with arbitrarily large
numerical values: note that what the cost function takes into account
is just whether a bound is violated or not, but not by how much. For
instance, in Example \ref{ex:running}, it could have been the case
that the \MAXSMT(QF-LIA) solver returned $w = y = 0$, $t = 1$,
$x = 10^5$, $v_{x^2} = 0$, etc. Since the model is used for extending
the domains, a large number would involve adding a prohibitive number
of case splitting clauses, and at the next iteration the
\MAXSMT(QF-LIA) solver would not able to handle the formula with a
reasonable amount of resources. However, having said that, as far as
we have been able to experiment, this kind of behaviour is rarely
observed in our implementation; see Section
\ref{sec:smt_nia_experiments_all} for more details. On the other hand, the
cost function in Section \ref{sec:omt_la} below does not suffer from
this drawback.

\subsection{An \OMT(QF-LIA) Approach to Domain Relaxation}
\label{sec:omt_la}

Another possibility of cost function for models of the linearization
is to measure the distance with respect to the artificial domains.
This can be cast as a problem in \OMT(QF-LIA) as follows.

Given a non-linear formula $F_0$, let us consider a linearization $F$
obtained after applying procedure \linearize with artificial
bounds $\bounds$. Now, let $\mathrm{vars}(\bounds)$ be the set of
variables $V$ for which an artificial domain
$[\l_V, \u_V] \in \bounds $ is added for the linearization. Formally,
the cost function is
$\sum_{V \in {\mathrm{vars}(\bounds)}} \delta(V, [\l_V, \u_V]),$
where $\delta(z, [\l, \u])$ is the \emph{distance} of $z$ with respect
to $[\l, \u]$:
\begin{center}
  $
  \delta(z, [\l, \u]) = \left\lbrace
    \begin{array}{ll}
      \l - z & \hbox{ if $z < \l$}\\
      0     & \hbox{ if $\l \leq z \leq \u$}\\
      z - \u & \hbox{ if $z > \u$}\\
    \end{array}
  \right.
  $
\end{center}
Note that, in the definition of the cost function, one can safely also
include bounds which are not artificial but derived from the
non-linear formula: the contribution to the cost of these is null,
since they are part of the original formula and therefore must always
be respected.

The approach is implemented in the procedure
\optimizeQFLIAOMT shown in Algorithm
\ref{alg:optimize_QF_LIA:omt}. In this procedure, an \OMT(QF-LIA)
solver is called (procedure \solveOMTQFLIA). Such a
system can be built upon an existing \SMT(QF-LIA) solver by adding an
optimization simplex phase II \cite{Schrijver} when the SAT engine
reaches a leaf of the search space. For the \OMT(QF-LIA) solver to
handle the cost function, the problem requires the following
reformulation. Let $\cost$ be the variable that the solver minimizes.
For each variable $V \in \mathrm{vars}(\bounds)$ with domain $[\l_V, \u_V]$, let us
introduce once and for all two extra integer variables $l_V$ and $u_V$
(meaning the distance with respect to the lower and to the upper bound
of the domain of $V$, respectively) and the \emph{auxiliary
  constraints} $l_V \geq 0 $, $l_V \geq \l_V - V $, $u_V \geq 0$,
$u_V \geq V - \u_V$. Then the cost function is determined by the
equation $\cost = \sum_{V \in \mathrm{vars}(\bounds)} (l_V + u_V)$, which is added to the
formula together with the aforementioned auxiliary constraints.

\begin{algorithm}[t]
\begin{lstlisting}
$\langle$ status, map $\rangle$ $\optimizeQFLIAOMT$(formula $F$, set $\bounds$) {

  $F' = F$;
  $E = 0$;                                               // expression for the cost function

  for ($V \geq L$ in $\bounds$) {
    $l_V$ = fresh_variable();
    $F'$ = $F'$ $\cup$ $\{ l_V \geq 0, l_V \geq L - V \}$;
    $E = E + l_V$;
  }

  for ($V \leq U$ in $\bounds$) {
    $u_V$ = fresh_variable();
    $F'$ = $F'$ $\cup$ $\{ u_V \geq 0, u_V \geq V - U \}$;
    $E = E + u_V$;
  }

  $F'$ = $F'$ $\cup$ {$\cost = E$};                                         // $\cost$ is the variable to be minimized

  return $\solveOMTQFLIA$($\langle \cost, F'\rangle)$);             // call to \omt solver
}
\end{lstlisting}
\caption{Procedure \optimizeQFLIAOMT based on \OMT(QF-LIA)}
\label{alg:optimize_QF_LIA:omt}
\end{algorithm}

The following result claims that the proposed cost function is admissible. Hence, by virtue of Theorem \ref{th:correctness},
if procedure \optimizeQFLIAOMT is implemented as in
Algorithm \ref{alg:optimize_QF_LIA:omt}, then
procedure \solveSMTQFNIAMinModels
is sound:

\begin{lemma}
\label{lemma:optimize_QF_LIA:omt}
  Let $F_0$ be an arbitrary formula in QF-NIA, and $F$ be any
  linearization of $F_0$ in QF-LIA obtained using the procedure
  \linearize with artificial bounds $B$.

  The function $\cost$ that takes as an input a model of $F$ and
  returns its distance to the artificial domains:

  $$ \sum_{V \in \mathrm{vars}(\bounds)} \delta(V, [\l_V, \u_V]) $$

  is admissible.
\end{lemma}

\begin{proof}
  The proof is analogous to that of Lemma
  \ref{lemma:optimize_QF_LIA:maxsmt}.
\end{proof}

Intuitively the proposed cost function corresponds to the \emph{number
  of new cases} that will have to be added in the next iteration of
the loop in \solveSMTQFNIAMinModels. However, it is also possible to
consider slightly different cost functions: for instance, one could
count the \emph{number of new clauses} that will have to be added. For
this purpose, it is only necessary to multiply variables $l_V$, $u_V$
in the equation that defines $\cost$ by the number of monomials that
were linearized by case splitting on $V$. In general, similarly to
Section \ref{sec:maxsmt_la}, one may have a template of cost function
of the form
$\cost = \sum_{V \in \mathrm{vars}(\bounds)} (\alpha_V\, l_V +
\beta_V\, u_V)$,
where $\alpha_V, \beta_V > 0$ for all $V \in \mathrm{vars}(\bounds)$.
Further, again these coefficients may be changed from one iteration to
the next one.

\begin{example}
  \label{ex:omt}
  Yet again let us take the same non-linear formula from Example
  \ref{ex:running}:
  $$tx + y \,\geq\, 4 \;\land\; t^2 + x^2 + w^2 + y^2 \,\leq\, 12.$$
  Let us also recall the artificial bounds:
  $-1 \leq t, x, w, y \leq 1$. By using the linearization $F$ as
  defined in Example \ref{ex:running}, one can express the resulting
  \OMT(QF-LIA) problem as follows:

  \begin{center}
    $\min \;
    \delta(t, [-1, 1]) +
    \delta(x, [-1, 1]) +
    \delta(w, [-1, 1]) +
    \delta(y, [-1, 1])\;\;$ subject to $F$,

  \end{center}
  or equivalently,
  \begin{center}
    $\min \;\cost\;\;$ subject to

    \medskip

    $F$ $\;\;\land$

    \smallskip
    $
    \begin{array}{rcl c rcl c rcl c rcl c rcl c rcl c rcl c cc}
      l_t  & \geq &  0 & \;\;\land\;\; & l_t  & \geq &  -1-t & \;\;\land\;\; & u_t  & \geq &  0 & \;\;\land\;\; & u_t  & \geq &  t-1 & \;\;\land \smallskip \\
      l_x  & \geq &  0 & \;\;\land\;\; & l_x  & \geq &  -1-x & \;\;\land\;\; & u_x  & \geq &  0 & \;\;\land\;\; & u_x  & \geq &  x-1 & \;\;\land \smallskip \\
      l_w  & \geq &  0 & \;\;\land\;\; & l_w  & \geq &  -1-w & \;\;\land\;\; & u_w  & \geq &  0 & \;\;\land\;\; & u_w  & \geq &  w-1 & \;\;\land \smallskip \\
      l_y  & \geq &  0 & \;\;\land\;\; & l_y  & \geq &  -1-y & \;\;\land\;\; & u_y  & \geq &  0 & \;\;\land\;\; & u_y  & \geq &  y-1 & \;\; \land \medskip \\
    \end{array}
    $

    $\cost = l_t + u_t + l_x + u_x + l_w  + u_w + l_y + u_y \;\;$\\

    \medskip

  \end{center}

  In this case, it can be seen that minimal solutions have cost $1$.
  For example, the \OMT(QF-LIA) solver could return the assignment:
  $x = v_{x^2} = 1$, $t = 2$, $v_{tx} = 4$ and
  $w = v_{w^2} = y = v_{y^2} = v_{t^2} = 0$. Note that, as $t = 2$ is
  not covered by the case splitting clauses, the values of $v_{tx}$
  and $v_{t^2}$ are unrelated to $t$. Now the new upper bound for $t$
  is $t \leq 2$, and clauses
\begin{center}
$
\begin{array}{lcl}
t = 2 & \limplies & v_{tx} = 2x\\
t = 2 & \limplies & v_{t^2} = 4\\
\end{array}
$
\end{center}
are added to the linearization.

At the next iteration there is still no solution with cost 0, so at
least another further iteration is necessary before a true model of
the non-linear formula can be found. \hfill$\blacksquare$
\end{example}

One of the drawbacks of this approach is that, as the previous example
suggests, domains may be enlarged very slowly. This implies that, in
cases where solutions have large numbers, many iterations are needed
before one of them is discovered. See Section
\ref{sec:smt_nia_experiments_all} below for more details on the performance
of this method in practice.

\subsection{Experimental Evaluation of  Model-guided Approaches}
\label{sec:smt_nia_experiments_all}

Here we evaluate experimentally our approaches for \SMT(QF-NIA) and
compare them with other non-linear solvers, namely those participating
in the QF-NIA division of the 2016 edition of SMT-COMP
(\url{http://smtcomp.sourceforge.net}). More in detail, we consider
the following tools:

\begin{itemize}

\item \aproveNia \cite{DBLP:conf/sat/FuhsGMSTZ07} with its default configuration;

\item \cvc4 \cite{DBLP:conf/cav/BarrettCDHJKRT11} version of 05-27-2016;

\item \prob \cite{DBLP:conf/sefm/KringsBL15};

\item \smtrat \cite{DBLP:conf/sat/CorziliusKJSA15};

\item \yices2 \cite{Dutertre:cav2014} version 2.4.2;

\item \rasat{} \cite{DBLP:conf/cade/TungKO16}, with two different versions:
  \rasat{-0.3} and \rasat{-0.4 exp};

\item \z3 \cite{MouraBjorner2008TACAS} version 4.4.1;

\item \barcelogicCores, our core-based algorithm \cite{Borrallerasetal2011JAR};

\item \barcelogicMaxsmt, our \maxsmt-based algorithm from Section \ref{sec:maxsmt_la};

\item \barcelogicOmt, our \omt-based algorithm from Section \ref{sec:omt_la}.

\end{itemize}

All \textsf{bcl-*} solvers\footnote{Available at
  \url{http://www.lsi.upc.edu/~albert/tocl2017.tgz}.} share
essentially the same underlying SAT engine and QF-LIA theory solver.
Moreover, some strategies are also common:

\begin{itemize}

\item procedure \artificialBounds uses a greedy algorithm
  for approximating the minimum set of variables that have to be
  introduced in the linearization (as shown in
  \cite{Borrallerasetal2011JAR}, computing a set with minimum size is
  NP-complete). For each of these variables we force the domain
  $[-1,1]$, even if variables have true bounds (for ease of
  presentation, we will assume here that true bounds always contain
  $[-1,1]$). This turns out to be useful in practice, as quite often
  satisfiable formulas have solutions with small coefficients. By
  forcing the domain $[-1,1]$, unnecessary case splitting clauses are
  avoided and the size of the linearized formula is reduced.

\item the first time a bound is chosen to be relaxed is handled
  specially. Let us assume it is the first time that a lower bound
  (respectively, an upper bound) of $V$ has to be relaxed. By virtue
  of the remark above, the bound must be of the form $V \geq -1$
  (respectively, $V \leq 1$). Now, if $V$ has a true bound of the form
  $V \geq L$ (respectively, $V \leq U$), then the new bound is the
  true bound. Otherwise, if $V$ does not have a true lower bound
  (respectively, upper bound), then the lower bound is decreased by
  one (respectively, the upper bound is increased by one). Again, this
  is useful to capture the cases in which there are solutions with
  small coefficients.

\item from the second time on, domain relaxation of \barcelogicMaxsmt
  and \barcelogicOmt follows basically what is described in Section
  \ref{sec:domain_relaxation}, except for a correction factor aimed at
  instances in which solutions have some large values. Namely, if
  $V \leq u$ has to be relaxed and in the minimal model $V$ is
  assigned value $U$, then the new upper bound is
  $U + \alpha \cdot \min(\beta,\frac{n}{m})$, where $\alpha$ and
  $\beta$ are parameters, $n$ is the number of times the upper bound
  of $V$ has been relaxed, and $m$ is the number of occurrences of $V$
  in the original formula. As regards \barcelogicCores, a similar
  expression is used in which the current bound $u$ is used instead of
  $U$, since there is no notion of ``best model''. The analogous
  strategy is applied for lower bounds.
\end{itemize}

The experiments were carried out on the StarExec
cluster~\cite{Starexec14}, whose nodes are equipped with Intel Xeon
2.4GHz processors. The memory limit was set to 60 GB, the same as in
the 2016 edition of SMT-COMP. As regards wall clock time, although in
SMT-COMP jobs were limited to 2400 seconds, in our experiments the
timeout was set to 1800 seconds, which is the maximum that StarExec
allowed us.

Two different sources of benchmarks were considered in this evaluation.
The first benchmark suite (henceforth referred to as \textsf{Term})
was already used in the conference version of this
paper~\cite{DBLP:conf/sat/LarrazORR14} and consists of 1934 instances
generated by the constraint-based termination prover described in
\cite{DBLP:conf/fmcad/LarrazORR13}. In these problems non-linear
monomials are quadratic.

The other benchmarks are the examples of QF-NIA in the SMT-LIB
\cite{BarFT-SMTLIB}, which are grouped into the following families:

\begin{itemize}
\item \textsf{AProVE}: 8829 instances
\item \textsf{calypto}: 177 instances
\item \textsf{LassoRanker}: 120 instances
\item \textsf{leipzig}: 167 instances
\item \textsf{mcm}: 186 instances
\item \textsf{UltimateAutomizer}: 7 instances
\item \textsf{UltimateLassoRanker}: 32 instances
\item \textsf{LCTES}: 2 instances
\end{itemize}

Results are displayed in two tables
(Tables~\ref{tab:smt_nia_experiments_all1} and
\ref{tab:smt_nia_experiments_all2}) for the sake of presentation. Rows
represent systems and distinguish between $\Sat$ and $\Unsat$
outcomes. Columns correspond to benchmark families. For each family,
the number of instances is indicated in parentheses. The cells either
show the number of problems of a given family that were solved by a
particular system with a certain outcome, or the total time (in
seconds) to process all of them. The best solver for each family (for
$\Sat$ and for $\Unsat$ examples) is highlighted in bold face.
%
\begin{center}
  \begin{table}%

  \tbl{Experimental evaluation of \SMT(QF-NIA) solvers on benchmark families \textsf{Term}, \textsf{AProVE}, \textsf{calypto} and \textsf{LassoRanker}.\label{tab:smt_nia_experiments_all1}}{%

  \begin{tabular}{|c|c||r|r||r|r||r|r||r|r|}
  \cline{3-10}

    \multicolumn{2}{c|}{} & \multicolumn{2}{c||}{\textsf{Term}}   & \multicolumn{2}{c||}{\textsf{AProVE}} & \multicolumn{2}{c||}{\textsf{calypto}} & \multicolumn{2}{c|}{\textsf{LassoRanker}} \\
    \multicolumn{2}{c|}{} & \multicolumn{2}{c||}{\textsf{(1934)}} & \multicolumn{2}{c||}{\textsf{(8829)}} & \multicolumn{2}{c||}{\textsf{(177)}}   & \multicolumn{2}{c|}{\textsf{(120)}} \\ \cline{3-10}
    \multicolumn{2}{c|}{} & \# p. & time & \# p. & time & \# p. & time & \# p. & time \\ \hline

\mr{ 2}{*}{\aproveNia}         &   \mr{ 1}{*}{$\Sat$}     &       0   &       0.00   &    8,028   &    4,242.65   &   77   &   1,715.02   &      3    &       1.97 \\ \cline{2-10}
                               &   \mr{ 1}{*}{$\Unsat$}   &       0   &       0.00   &        0   &        0.00   &    0   &       0.00   &      0    &       0.00 \\ \hline\hline
\mr{ 2}{*}{\cvc4}              &   \mr{ 1}{*}{$\Sat$}     &      45   &   8,898.80   &    7,892   &  144,767.87   &   25   &      78.29   &      4    &     692.98 \\ \cline{2-10}
                               &   \mr{ 1}{*}{$\Unsat$}   &       0   &       0.00   &       10   &        0.18   &   35   &       0.71   &     71    &       1.26 \\ \hline\hline
\mr{ 2}{*}{\prob}              &   \mr{ 1}{*}{$\Sat$}     &       0   &       0.00   &    7,415   &   19,715.86   &   41   &      85.52   &      3    &       3.07 \\ \cline{2-10}
                               &   \mr{ 1}{*}{$\Unsat$}   &       0   &       0.00   &       16   &       15.70   &   13   &     498.51   &      0    &       0.00 \\ \hline\hline
\mr{ 2}{*}{\smtrat}            &   \mr{ 1}{*}{$\Sat$}     &     232   &  82,122.64   &    8,026   &      313.44   &   79   &     163.58   &      3    &       0.64 \\ \cline{2-10}
                               &   \mr{ 1}{*}{$\Unsat$}   &      15   &   1,377.74   &      221   &    7,654.78   &   89   &     663.89   &     21    &      12.59 \\ \hline\hline
\mr{ 2}{*}{\yices2}            &   \mr{ 1}{*}{$\Sat$}     &   1,830   &  79,764.09   &    7,959   &    3,293.65   &   79   &       6.53   &      4    &       0.16 \\ \cline{2-10}
                               &   \mr{ 1}{*}{$\Unsat$}   &      69   &     940.15   &   \b{764}  & \b{4,964.66}  &\b{97}  &  \b{488.38}  &     97    &     875.44 \\ \hline\hline
\mr{ 2}{*}{\rasat{-0.3}}       &   \mr{ 1}{*}{$\Sat$}     &      20   &   2,444.87   &    7,421   &   35,053.18   &   32   &   3,393.93   &      3    &       2.41 \\ \cline{2-10}
                               &   \mr{ 1}{*}{$\Unsat$}   &       0   &       0.00   &      320   &  554,482.86   &   47   &  30,232.16   &     43    &  75,603.23 \\ \hline\hline
\mr{ 2}{*}{\rasat{-0.4 exp}}   &   \mr{ 1}{*}{$\Sat$}     &      36   &   5,161.97   &    7,745   &   50,695.06   &   31   &     954.16   &      3    &       1.54 \\ \cline{2-10}
                               &   \mr{ 1}{*}{$\Unsat$}   &       4   &   2,454.21   &       18   &      105.59   &   31   &     547.26   &      2    &       2.46 \\ \hline\hline
\mr{ 2}{*}{\z3}                &   \mr{ 1}{*}{$\Sat$}     &     194   &  77,397.16   &    8,023   &   14,790.21   &   79   &     943.03   &      4    &      13.16 \\ \cline{2-10}
                               &   \mr{ 1}{*}{$\Unsat$}   &   \b{70}  &\b{3,459.77}  &      286   &    7,989.62   &   96   &   1,932.11   &    100    &   3,527.34 \\ \hline\hline
\mr{ 2}{*}{\barcelogicCores}   &   \mr{ 1}{*}{$\Sat$}     &   1,857   &   4,396.09   & \b{8,028}  & \b{1,726.49}  &   80   &       6.20   &      4    &       0.09 \\ \cline{2-10}
                               &   \mr{ 1}{*}{$\Unsat$}   &       0   &       0.00   &       15   &        0.41   &   94   &   1,596.99   &     72    &       2.53 \\ \hline\hline
\mr{ 2}{*}{\barcelogicMaxsmt}  &   \mr{ 1}{*}{$\Sat$}     &\b{1,857}  &  \b{811.54}  &    8,027   &    1,763.70   &\b{80}  &    \b{5.74}  &   \b{4}   &    \b{0.08}\\ \cline{2-10}
                               &   \mr{ 1}{*}{$\Unsat$}   &      67   &      31.33   &      202   &       51.50   &   97   &     994.17   & \b{103}   &    \b{2.96}\\ \hline\hline
\mr{ 2}{*}{\barcelogicOmt}     &   \mr{ 1}{*}{$\Sat$}     &   1,854   &   6,420.59   &    8,013   &   25,274.94   &   80   &       6.75   &      4    &       0.10 \\ \cline{2-10}
                               &   \mr{ 1}{*}{$\Unsat$}   &      67   &      34.99   &      203   &       36.18   &   97   &   1,327.95   &    103    &       3.59 \\ \hline
  \end{tabular}
}
\end{table}
\end{center}
\begin{center}
  \begin{table}%

  \tbl{Experimental evaluation of \SMT(QF-NIA) solvers on benchmark families \textsf{leipzig}, \textsf{mcm}, \textsf{UltimateAutomizer (UA)} and \textsf{UltimateLassoRanker (ULR)}.\label{tab:smt_nia_experiments_all2}}{%

  \begin{tabular}{|c|c||r|r||r|r||r|r||r|r|}
  \cline{3-10}

    \multicolumn{2}{c|}{}  & \multicolumn{2}{c||}{\textsf{leipzig}} & \multicolumn{2}{c||}{\textsf{mcm}}   & \multicolumn{2}{c||}{\textsf{UA}} & \multicolumn{2}{c|}{\textsf{ULR}} \\
    \multicolumn{2}{c|}{}  & \multicolumn{2}{c||}{\textsf{(167)}}   & \multicolumn{2}{c||}{\textsf{(186)}} & \multicolumn{2}{c||}{\textsf{(7)}}               & \multicolumn{2}{c|}{\textsf{(32)}} \\ \cline{3-10}

    \multicolumn{2}{c|}{} & \# p. & time & \# p. & time & \# p. & time & \# p. & time \\ \hline

\mr{ 2}{*}{\aproveNia}         &   \mr{ 1}{*}{$\Sat$}     &   161   &   1,459.27   &    0   &          0.00   &   0   &   0.00   &    6   &      5.02 \\ \cline{2-10}
                               &   \mr{ 1}{*}{$\Unsat$}   &     0   &       0.00   &    0   &          0.00   &   0   &   0.00   &    0   &      0.00 \\ \hline\hline
\mr{ 2}{*}{\cvc4}              &   \mr{ 1}{*}{$\Sat$}     &\b{162}  &  \b{237.63}  &\b{48}  &  \b{22,899.02}  &   0   &   0.00   &    6   &      3.76 \\ \cline{2-10}
                               &   \mr{ 1}{*}{$\Unsat$}   &     0   &       0.00   &    0   &          0.00   &   6   &   0.06   &   22   &     69.19 \\ \hline\hline
\mr{ 2}{*}{\prob}              &   \mr{ 1}{*}{$\Sat$}     &    50   &      54.81   &    1   &      1,631.89   &   0   &   0.00   &    4   &      5.58 \\ \cline{2-10}
                               &   \mr{ 1}{*}{$\Unsat$}   &     0   &       0.00   &    0   &          0.00   &   1   &   1.02   &    1   &      1.34 \\ \hline\hline
\mr{ 2}{*}{\smtrat}            &   \mr{ 1}{*}{$\Sat$}     &   160   &   2,827.37   &   21   &      2,516.21   &   0   &   0.00   &    6   &      0.86 \\ \cline{2-10}
                               &   \mr{ 1}{*}{$\Unsat$}   &     0   &       0.00   &    0   &          0.00   &   1   &   2.44   &   24   &    186.14 \\ \hline\hline
\mr{ 2}{*}{\yices2}            &   \mr{ 1}{*}{$\Sat$}     &    92   &     715.04   &   11   &      5,816.44   &   0   &   0.00   & \b{6}  &   \b{0.05} \\ \cline{2-10}
                               &   \mr{ 1}{*}{$\Unsat$}   &  \b{1}  &    \b{0.01}  &    0   &          0.00   &\b{7}  &\b{0.02}  &\b{26}  &  \b{11.07}\\ \hline\hline
\mr{ 2}{*}{\rasat{-0.3}}       &   \mr{ 1}{*}{$\Sat$}     &    32   &  15,758.07   &    2   &      1,787.57   &   0   &   0.00   &    2   &      5.88 \\ \cline{2-10}
                               &   \mr{ 1}{*}{$\Unsat$}   &     1   &   1,800.07   &\b{99}  & \b{178,204.54}  &   1   &   5.28   &    1   &  1,351.68 \\ \hline\hline
\mr{ 2}{*}{\rasat{-0.4 exp}}   &   \mr{ 1}{*}{$\Sat$}     &   134   &  17,857.21   &    8   &      3,309.13   &   0   &   0.00   &    3   &      1.60 \\ \cline{2-10}
                               &   \mr{ 1}{*}{$\Unsat$}   &     0   &       0.00   &    0   &          0.00   &   1   &   8.08   &    1   &      1.50 \\ \hline\hline
\mr{ 2}{*}{\z3}                &   \mr{ 1}{*}{$\Sat$}     &   162   &   1,472.00   &   23   &      3,906.84   &   0   &   0.00   &    6   &      0.34 \\ \cline{2-10}
                               &   \mr{ 1}{*}{$\Unsat$}   &     0   &       0.00   &    7   &      7,127.61   &   7   &   0.54   &   26   &     45.20 \\ \hline\hline
\mr{ 2}{*}{\barcelogicCores}   &   \mr{ 1}{*}{$\Sat$}     &   158   &   3,596.74   &   15   &      1,160.10   &   0   &   0.00   &    6   &      0.33 \\ \cline{2-10}
                               &   \mr{ 1}{*}{$\Unsat$}   &     0   &       0.00   &    0   &          0.00   &   1   &   0.06   &   24   &     32.87 \\ \hline\hline
\mr{ 2}{*}{\barcelogicMaxsmt}  &   \mr{ 1}{*}{$\Sat$}     &   153   &   4,978.91   &   17   &      1,004.25   &   0   &   0.00   &    6   &      0.31 \\ \cline{2-10}
                               &   \mr{ 1}{*}{$\Unsat$}   &     0   &       0.00   &    0   &          0.00   &   1   &   0.02   &   26   &     28.56 \\ \hline\hline
\mr{ 2}{*}{\barcelogicOmt}     &   \mr{ 1}{*}{$\Sat$}     &   148   &   7,351.45   &   19   &      2,937.99   &   0   &   0.00   &    6   &      0.34 \\ \cline{2-10}
                               &   \mr{ 1}{*}{$\Unsat$}   &     0   &       0.00   &    0   &          0.00   &   1   &   0.02   &   26   &     29.36 \\ \hline

  \end{tabular}
}
\end{table}
\end{center}
Due to lack of space, the results for family \textsf{LCTES} do not
appear in the tables. This family consists of just two benchmarks, one
of which was not solved by any system. The other instance was only
solved by \cvc4, which reported $\Unsat$ in 0.5 seconds.

As the tables indicate, overall our techniques perform well on
$\Sat$ instances, being the results particularly favourable for the
\textsf{Term} family. This is natural: linearizing by case splitting
is aimed at finding solutions quickly without having to pay the toll
of heavy-weight non-linear reasoning. If satisfiable instances have
solutions with small domains (which is often the case, for instance,
when they come from our program analysis applications), our techniques
usually work well. On the other hand, for families \textsf{Aprove},
\textsf{leipzig} and \textsf{mcm} the results are only comparable or
slightly worse than those obtained with other tools\footnote{However,
  it must be remarked that we detected several inconsistencies between
  \rasat{-0.3} and the rest of the solvers in the family \textsf{mcm},
  which makes the results of this tool unreliable.}. One of the reasons
could be that, at least for \textsf{Aprove} and \textsf{leipzig},
formulas have a very simple Boolean structure: they are essentially
conjunctions of literals and few clauses (if any). For this particular
kind of problems CAD-based techniques such as those implemented in
\yices2 and \z3, which are precisely targeted at conjunctions of
non-linear literals, may be more adequate.

Regarding $\Unsat$ instances, it can be seen that our approaches,
while often competitive, can be outperformed by other tools in some
families. Again, this is not surprising: linearizing may not be
sufficient to detect unsatisfiability when deep non-linear reasoning
is required. On the other hand, sometimes there may be a purely linear
argument that proves that an instance is unsatisfiable. Our techniques
can be effective in these situations, which may be relatively frequent
depending on the application. This would be the case of families
\textsf{Term}, \textsf{calypto}, \textsf{LassoRanker} and
\textsf{ULR}.

Comparing our techniques among themselves, overall \barcelogicMaxsmt
tends to give the best results in terms of number of solved $\Sat$ and
$\Unsat$ instances and timings. For example, we can see that
\barcelogicCores proves many fewer unsatisfiable instances than
model-guided approaches. The reason is the following. Let $F_0$ be a
formula in QF-NIA, and $F$ be a linearization of $F_0$ computed with
artificial bounds $B$.
Let us assume that $F$ is unsatisfiable. In this case, when the
algorithm in \barcelogicCores tests the satisfiability of $F \land B$,
it finds that it is unsatisfiable. Then, if we are lucky and an
unsatisfiable core that only uses clauses from $F$ is obtained, then
it can be concluded that $F_{0}$ is unsatisfiable immediately.
However, there may be other unsatisfiable cores of $F \land B$, which
use artificial bounds \footnote{For the sake of efficiency,
  \barcelogicCores does not guarantee that cores are minimal with
  respect to subset inclusion: computing \emph{minimal unsatisfiable
    sets} \cite{DBLP:journals/aicom/BelovLM12} to eliminate irrelevant
  clauses implies an overhead that in our experience does not pay off.
  But even if minimality were always achieved, there could still be
  unsatisfiable cores in $F \land B$ using artificial bounds.}. Using
such a core leads to performing yet another (useless) iteration of
domain relaxation. Unfortunately the choice of the unsatisfiable core
depends on the way the search space is explored, which does not take
into account whether bounds are original or artificial so as not to
interfere with the Boolean engine heuristics. On the other hand,
model-guided approaches always detect when the linearization is
unsatisfiable.
As for $\Sat$ instances, the number of solved problems of
\barcelogicCores is similar to that of \barcelogicMaxsmt, but the
latter tends to be faster.

Regarding \barcelogicOmt, it turns out that, in general, the
additional iterations required in the simplex algorithm to perform the
optimization are too expensive. Moreover, after inspecting the traces
we have confirmed that as Example \ref{ex:omt} suggested, \barcelogicOmt
enlarges the domains too slowly, which is hindering the search.

\subsection{Variants}
\label{sec:variants}

According to the experiments in Section
\ref{sec:smt_nia_experiments_all}, altogether the approach based on
\MAXSMT(QF-LIA) gives the best results among our methods. In this
section we propose several ideas for improving it further.

\subsubsection{Non-incremental Strategy}
\label{sec:non-incremental}

A common feature of the procedures for solving \SMT(QF-NIA) described
in Sections \ref{sec:unsatisfiable_cores}, \ref{sec:maxsmt_la} and
\ref{sec:omt_la} is that, when no model of the linearization is found
that satisfies all artificial bounds, the domains are enlarged. Thus,
iteration after iteration, the number of case splitting clauses
increases. In this sense, the aforementioned methods are
\emph{incremental}. A disadvantage of this incrementality is that,
after some iterations, the formula size may easily become too big to
be manageable.

On the other hand, instead of enlarging a domain, one can follow a
\emph{non-incremental} strategy and \emph{replace} the domain by another one
that might not include it. For example, in the model-guided
approaches, when computing the new domain for a variable one may
discard the current domain and for the next iteration take an interval
centered at the value in the minimal model (procedure
\relaxDomainsMinModelsNonInc, described in Algorithm
\ref{alg:relax_domains_min_models_non_inc}). The procedure for
updating the formula has to be adapted accordingly too, so that the
case splitting clauses correspond to the values in the artificial
domains (procedure \updateNonInc, shown in Algorithm
\ref{alg:update_non_inc}). In this fashion one can control precisely
the number of case splitting clauses, and therefore the size of the
formula.

\begin{algorithm}[t]
\begin{lstlisting}
set $\relaxDomainsMinModelsNonInc$(set $\bounds$, map $\model$) { // returns the new set of artificial bounds
  let $S \subseteq \{b \;|\; b \in \bounds, \model \not\models b\}$ such that $S \neq \emptyset$;               // choose among bounds violated by the model
  $W = \{ \textit{var}(b) \st b \in S\}$;                                                      // variables with domain to be relaxed
  $\bounds' = \{ b \st b \in \bounds \land \textit{var}(b) \not\in W \}$;
  for ($V$ in $W$)
    $\bounds'$ = $\bounds'$ $\cup$ {$\model(V) - R \leq V \leq \model(V) + R$};                // for a parameter $R > 0$
  return $\bounds'$;
}
\end{lstlisting}
\caption{Procedure \relaxDomainsMinModelsNonInc}
\label{alg:relax_domains_min_models_non_inc}
\end{algorithm}

\begin{algorithm}[t]
\begin{lstlisting}
formula $\updateNonInc$(formula $F$, set $\bounds$, set $\bounds'$) {
  $W$ = $\{ \textit{var}(b) \st b \in \bounds' - \bounds\}$;                                                     // variables whose domain has changed
  $F'$ = remove_case_splits($F$, $W$);                                  // remove outdated case splitting clauses
  for ($V$ in $W$) {
    let $L$, $U$ such that $L \leq V\leq U \in \bounds'$;
    for ($K$ in [$L$, $U$])
       for ($Q$ such that $V$ == linearization_variable($Q$))      // $V$ was used to linearize monomial $Q$
          $F'$ = $F'$ $\cup$ { $V = K \limplies V_Q =$ evaluate($Q$, $V$, $K$)};    // $V_Q$ is the variable standing for $Q$
  }
  return $F' \cup \{\bigvee_{b \in \bounds} \lnot b \}$;                                                    // forbid $\bounds$: no solution there
}
\end{lstlisting}
  \caption{Procedure \updateNonInc}
  \label{alg:update_non_inc}
\end{algorithm}

Since monotonicity of domains from one iteration to the next one is
now not maintained, this approach requires bookkeeping so as to avoid
repeating the same choice of artificial domains. One way to implement
this is to add clauses that forbid each of the combinations of domains
that have already been tried and with which no model of the original
formula was found. Namely, let $\bounds$ be such a combination of
artificial bounds. We add the (hard) clause
$\lor_{b \in \bounds} \lnot b$, which forces that at least one of the
bounds in $\bounds$ cannot hold. Now, in any following iteration, if
the minimal model $M$ of the linearization does not satisfy the
non-linear formula, the new set of artificial bounds $\bounds'$ is
such that $M \models \bounds'$. This together with
$M \not \models \bounds$ implies $\bounds' \not \models \bounds$, and
therefore artificial bounds cannot be repeated. See procedure
$\updateNonInc$ in Algorithm \ref{alg:update_non_inc}.

Finally note that, although this alternative strategy for producing
new artificial bounds can in principle be adapted to either of the
model-guided methods, it makes the most sense for the
\MAXSMT(QF-LIA)-based procedure. The reason is that, being
model-guided, in this approach the next domains to be considered are
determined by the minimal model and, as already observed in Section
\ref{sec:maxsmt_la}, this minimal model may assign large values to
variables and thus lead to intractable formula growth.

\begin{example}
\label{ex:noninc}
Let us take the formula and artificial bounds of the running example.
We resume Example \ref{ex:maxsmt}, where the following minimal
solution of cost $1$ was shown: $t = v_{t^2} = 1$, $x = v_{tx} = 4$
and $w = v_{w^2} = y = v_{y^2} = v_{x^2} = 0$, being $x \leq 1$ the
only violated artificial bound. Now, taking a radius $R = 2$ for the
interval around $x = 4$, in the next iteration the following
artificial bounds would be considered: $-1 \leq t, w, y \leq 1$ and
$2 \leq x \leq 6$. Moreover, the following clause would be added to
the linearization:
  $$
  -1 > t \;\;\lor\;\; 1 < t \;\;\lor\;\;
  -1 > x \;\;\lor\;\; 1 < x \;\;\lor\;\;
  -1 > w \;\;\lor\;\; 1 < w \;\;\lor\;\;
  -1 > y \;\;\lor\;\; 1 < y
  $$
  together with
$$
  \begin{array}{lcl c lcl c}
( x = 2 & \limplies & v_{x^2} =  4  ) & \;\;\;\land\\
( x = 3 & \limplies & v_{x^2} =  9  ) & \;\;\;\land\\
( x = 4 & \limplies & v_{x^2} =  16 ) & \;\;\;\land\\
( x = 5 & \limplies & v_{x^2} =  25 ) & \;\;\;\land\\
( x = 6 & \limplies & v_{x^2} =  36 ), & \\
  \end{array}
  $$
  while clauses
$$
  \begin{array}{lcl c lcl c}
( x =-1 & \limplies & v_{x^2} =  1 ) & \;\;\;\land\\
( x = 0 & \limplies & v_{x^2} =  0 ) & \;\;\;\land\\
( x = 1 & \limplies & v_{x^2} =  1 ) & \\
  \end{array}
  $$
would be removed.\hfill$\blacksquare$
\end{example}

\subsubsection{Optimality Cores}
\label{sec:optimality-cores}

When following the approach presented in Section
\ref{sec:non-incremental}, one needs to keep track of the combinations
of domains that have already been attempted, in order to avoid
repeating work and possibly entering into cycles. As pointed out
above, this can be achieved for instance by adding clauses that
exclude these combinations of domains. From the SMT perspective, these
clauses can be viewed as \emph{conflict explanations}, if one
understands a conflict as a choice of artificial domains that does not
lead to a solution to the original non-linear problem. Following the
SMT analogy, it is important that explanations are as short as
possible. In this section we present a technique aimed at reducing the
size of these explanations.

Let us focus on the model-guided procedures, in which at each
iteration one minimizes a cost function.

\begin{definition}
  Let $F$ be a clause set with variables $X$, and let
  $\cost: X \goes \reals$ be a function.
  Let $C \subseteq F$ be a subset of the clauses and
  $\cost': X \goes \reals$ be another function.
  We say that the pair $(C, cost')$ is an
  \emph{optimality core} of the minimization problem
  $\min\{ \cost(\model) \st \model \models F \}$ if the following conditions hold:

  \begin{enumerate}

  \item $cost' \leq cost$; and

  \item $\min\{ \cost'(\model) \st \model \models C \} = \min\{ \cost(\model) \st \model \models F \}.$
  \end{enumerate}

\end{definition}

The interest of optimality cores in our context is that they allow
identifying a subset of the artificial bounds that is enough for
soundly discarding the current combination of domains (and possibly
many more). Let us now describe how optimality cores and these subsets
of artificial bounds may be obtained in the \maxsmt approach. In this
case, one searches the minimum number of violated artificial bounds.
In a similar way to resolution refutations obtained from unsatisfiable
instances, after each call to the \MAXSMT(QF-LIA) solver on the
linearization with soft artificial bounds one may retrieve a lower
bound certificate \cite{DBLP:journals/jar/LarrosaNOR11}. This
certificate consists essentially of a tree of \emph{cost resolution}
steps, and proves that any model of the linearization will violate at
least as many artificial bounds as the reported optimal model. Now
consider the leaves of this tree. Let $H$ be those leaves which are
hard clauses (corresponding to clauses of the linearization), and $S$
be those which are soft (corresponding to soft bounds). If we call
$\cost_S$ the function that counts the number of bounds from $S$ that
are not satisfied, from the definition it follows that $(H, \cost_S)$
is an optimality core. In particular, we are interested in the set of
artificial bounds $S$: if the optimal model of the \maxsmt problem has
positive cost, then there is no model of the linearization that can
also satisfy $S$. So we can soundly add the clause
$\lor_{b \in S} \lnot b$, which replaces the (potentially much) longer
clause $ \lor_{b \in \bounds} \lnot b$ in Algorithm
\ref{alg:update_non_inc}.

\begin{example}
  Once more let us consider the running example. We proceed as in
  Example \ref{ex:noninc}, but instead of adding the clause
  $$
  -1 > t \;\;\lor\;\; 1 < t \;\;\lor\;\;
  -1 > x \;\;\lor\;\; 1 < x \;\;\lor\;\;
  -1 > w \;\;\lor\;\; 1 < w \;\;\lor\;\;
  -1 > y \;\;\lor\;\; 1 < y
  $$
  we add
  $$
  -1 > t \;\;\lor\;\; 1 < t \;\;\lor\;\;
  -1 > x \;\;\lor\;\; 1 < x \;\;\lor\;\;
  1 < y,
  $$
  i.e., we discard the literals $-1 > w$,\, $1 < w$ and $-1 > y$.
  Following the above notation, we can do that since $(H,\cost_S)$ is an
  optimality core, where $H$ is
$$ \{  v_{tx} + y \geq 4,\quad
t =-1 \limplies v_{tx} = -x, \quad
t = 0 \limplies v_{tx} =  0, \quad
t = 1 \limplies v_{tx} =  x \}
$$
and $S$ is $$\{ [-1 \leq t, 1],\quad [t \leq 1, 1],\quad [-1 \leq x, 1],\quad [x \leq 1, 1],\quad [y \leq 1, 1] \}.$$
\hfill$\blacksquare$
\end{example}

\subsection{Experimental Evaluation of  \MAXSMT(QF-LIA)-based Approaches}
\label{sec:smt_nia_experiments_maxsmt}

In this section we evaluate experimentally the variations of the
\MAXSMT(QF-LIA) approach proposed in Sections
\ref{sec:non-incremental} and \ref{sec:optimality-cores}. In addition
to the benchmarks used in Section \ref{sec:smt_nia_experiments_all},
we have also considered instances produced by our constraint-based
termination prover \tool{VeryMax}
(\url{http://www.cs.upc.edu/~albert/VeryMax.html}) on the divisions of
the termination competition termCOMP 2016~\cite{TermComp} in which it
participated, namely \textsf{Integer Transition Systems} and
\textsf{C Integer}. Since internally \tool{VeryMax} generates
\MAXSMT(QF-NIA) rather than \SMT(QF-NIA) problems, soft clauses were
removed. Given the huge number of obtained examples, of the order of
tens of thousands, we could not afford carrying out the experiments
will all the tools considered in Section
\ref{sec:smt_nia_experiments_all}, and had to restrict the evaluation
to the competing solvers that overall performed the best, namely \z3
and \yices2. Hence, in addition to these two, the following solvers
are considered here:

\begin{itemize}
\item \barcelogicMaxsmt, the \MAXSMT(QF-LIA)-based approach as in
  Section \ref{sec:smt_nia_experiments_all};
\item \barcelogicJmpNoCores, the non-incremental algorithm from Section \ref{sec:non-incremental};
\item \barcelogicJmpCores, the non-incremental algorithm that uses
  optimality cores from Section \ref{sec:optimality-cores}.
\end{itemize}

Moreover, to further reduce the time required by the experiments, we
decided to discard those benchmarks which could be solved both by
\yices2 and \barcelogicMaxsmt in negligible time (less than 0.5
seconds). After this filtering, finally 20354 and 2019 benchmarks were
included in families \textsf{Integer Transition Systems} and \textsf{C
  Integer}, respectively.

Results are displayed in Tables~\ref{tab:smt_nia_experiments_sel1},
\ref{tab:smt_nia_experiments_sel2} and
\ref{tab:smt_nia_experiments_sel3}, following the same format as in
Section \ref{sec:smt_nia_experiments_all}.

\begin{center}
  \begin{table}%

  \tbl{Experimental evaluation of \SMT(QF-NIA) solvers on benchmark families \textsf{Term}, \textsf{AProVE}, \textsf{calypto} and \textsf{LassoRanker}.\label{tab:smt_nia_experiments_sel1}}{%

  \begin{tabular}{|c|c||r|r||r|r||r|r||r|r|}
  \cline{3-10}

    \multicolumn{2}{c|}{} & \multicolumn{2}{c||}{\textsf{Term}}   & \multicolumn{2}{c||}{\textsf{AProVE}} & \multicolumn{2}{c||}{\textsf{calypto}} & \multicolumn{2}{c|}{\textsf{LassoRanker}} \\
    \multicolumn{2}{c|}{} & \multicolumn{2}{c||}{\textsf{(1934)}} & \multicolumn{2}{c||}{\textsf{(8829)}} & \multicolumn{2}{c||}{\textsf{(177)}}   & \multicolumn{2}{c|}{\textsf{(120)}} \\ \cline{3-10}
    \multicolumn{2}{c|}{} & \# p. & time & \# p. & time & \# p. & time & \# p. & time \\ \hline

\mr{2}{*}{\yices2}               & \mr{1}{*}{$\Sat$}     &    1,830   &  79,764.09   &    7,959   &     3,293.65   &   79   &       6.53   &      4  &       0.16  \\ \cline{2-10}
                                 & \mr{1}{*}{$\Unsat$}   &       69   &     940.15   &   \b{764}  &  \b{4,964.66}  &   97   &     488.38   &     97  &     875.44  \\ \hline\hline
\mr{2}{*}{\z3}                   & \mr{1}{*}{$\Sat$}     &      194   &  77,397.16   &    8,023   &    14,790.21   &   79   &     943.03   &      4  &      13.16  \\ \cline{2-10}
                                 & \mr{1}{*}{$\Unsat$}   &    \b{70}  &\b{3,459.77}  &      286   &     7,989.62   &   96   &   1,932.11   &    100  &   3,527.34  \\ \hline\hline
\mr{2}{*}{\barcelogicMaxsmt}     & \mr{1}{*}{$\Sat$}     &    1,857   &     811.54   &    8,027   &     1,763.70   &   80   &       5.74   &   \b{4} &    \b{0.08} \\ \cline{2-10}
                                 & \mr{1}{*}{$\Unsat$}   &       67   &      31.33   &      202   &        51.50   &   97   &     994.17   & \b{103} &    \b{2.96} \\ \hline\hline
\mr{2}{*}{\barcelogicJmpNoCores} & \mr{1}{*}{$\Sat$}     & \b{1,857}  &  \b{276.20}  & \b{8,028}  &  \b{1,777.97}  &\b{80}  &    \b{5.50}  &      4  &       0.10  \\ \cline{2-10}
                                 & \mr{1}{*}{$\Unsat$}   &       67   &     191.66   &      202   &        51.72   &\b{97}  &  \b{155.32}  &    103  &      13.63  \\ \hline\hline
\mr{2}{*}{\barcelogicJmpCores}   & \mr{1}{*}{$\Sat$}     &    1,857   &     349.48   &    8,028   &     1,825.93   &   80   &       5.54   &      4  &       0.10  \\ \cline{2-10}
                                 & \mr{1}{*}{$\Unsat$}   &       67   &     184.41   &      202   &        57.52   &   97   &     273.15   &    103  &       9.31  \\ \hline
  \end{tabular}
}
\end{table}
\end{center}
\begin{center}
  \begin{table}%

  \tbl{Experimental evaluation of \SMT(QF-NIA) solvers on benchmark families \textsf{leipzig}, \textsf{mcm}, \textsf{UltimateAutomizer (UA)} and \textsf{UltimateLassoRanker (ULR)}.\label{tab:smt_nia_experiments_sel2}}{%

  \begin{tabular}{|c|c||r|r||r|r||r|r||r|r|}
  \cline{3-10}

    \multicolumn{2}{c|}{}  & \multicolumn{2}{c||}{\textsf{leipzig}} & \multicolumn{2}{c||}{\textsf{mcm}}   & \multicolumn{2}{c||}{\textsf{UA}} & \multicolumn{2}{c|}{\textsf{ULR}} \\
    \multicolumn{2}{c|}{}  & \multicolumn{2}{c||}{\textsf{(167)}}   & \multicolumn{2}{c||}{\textsf{(186)}} & \multicolumn{2}{c||}{\textsf{(7)}}               & \multicolumn{2}{c|}{\textsf{(32)}} \\ \cline{3-10}

    \multicolumn{2}{c|}{} & \# p. & time & \# p. & time & \# p. & time & \# p. & time \\ \hline

\mr{2}{*}{\yices2}               & \mr{1}{*}{$\Sat$}     &    92   &     715.04   &   11   &    5,816.44   &   0   &   0.00   & \b{6}  &   \b{0.05}\\ \cline{2-10}
                                 & \mr{1}{*}{$\Unsat$}   &  \b{1}  &    \b{0.01}  &    0   &        0.00   &\b{7}  &\b{0.02} & \b{26}  &  \b{11.07} \\ \hline\hline
\mr{2}{*}{\z3}                   & \mr{1}{*}{$\Sat$}     &\b{162}  &\b{1,472.00}  &\b{23}  & \b{3,906.84}  &   0   &   0.00   &    6   &      0.34 \\ \cline{2-10}
                                 & \mr{1}{*}{$\Unsat$}   &     0   &       0.00   & \b{7}  & \b{7,127.61}  &   7   &   0.54   &   26   &     45.20 \\ \hline\hline
\mr{2}{*}{\barcelogicMaxsmt}     & \mr{1}{*}{$\Sat$}     &   153   &   4,978.91   &   17   &    1,004.25   &   0   &   0.00   &    6   &      0.31 \\ \cline{2-10}
                                 & \mr{1}{*}{$\Unsat$}   &     0   &       0.00   &    0   &        0.00   &   1   &   0.02   &   26   &     28.56 \\ \hline\hline
\mr{2}{*}{\barcelogicJmpNoCores} & \mr{1}{*}{$\Sat$}     &   155   &   5,193.20   &   23   &    3,983.74   &   0   &   0.00   &    6   &      0.33 \\ \cline{2-10}
                                 & \mr{1}{*}{$\Unsat$}   &     0   &       0.00   &    0   &        0.00   &   1   &   0.02   &   26   &     28.47 \\ \hline\hline
\mr{2}{*}{\barcelogicJmpCores}   & \mr{1}{*}{$\Sat$}     &   156   &   3,602.32   &   19   &    2,037.77   &   0   &   0.00   &    6   &      0.40 \\ \cline{2-10}
                                 & \mr{1}{*}{$\Unsat$}   &     0   &       0.00   &    0   &        0.00   &   1   &   0.02   &   26   &     31.28 \\ \hline
  \end{tabular}
}
\end{table}
\end{center}
\begin{center}
  \begin{table}%

    \tbl{Experimental evaluation of \SMT(QF-NIA) solvers on benchmark families \textsf{LCTES}, \textsf{Integer Transition Systems (ITS)} and \textsf{C Integer (CI)}.\label{tab:smt_nia_experiments_sel3}}{%

  \begin{tabular}{|c|c||r|r||r|r||r|r|}
  \cline{3-8}

    \multicolumn{2}{c|}{}  & \multicolumn{2}{c||}{\textsf{LCTES}} & \multicolumn{2}{c||}{\textsf{ITS}}   & \multicolumn{2}{c|}{\textsf{CI}} \\
    \multicolumn{2}{c|}{}  & \multicolumn{2}{c||}{\textsf{(2)}}   & \multicolumn{2}{c||}{\textsf{(20354)}} & \multicolumn{2}{c|}{\textsf{(2019)}} \\ \cline{3-8}

    \multicolumn{2}{c|}{} & \# p. & time & \# p. & time & \# p. & time \\ \hline

\mr{2}{*}{\yices2}                &  \mr{1}{*}{$\Sat$}    &  0  &  0.00  &    8,408   &   471,160.33   &   714   &    84,986.50   \\ \cline{2-8}
                                  &  \mr{1}{*}{$\Unsat$}  &  0  &  0.00  & \b{4,085}  &\b{142,965.19}  &   246   &    24,498.79   \\ \hline\hline
\mr{2}{*}{\z3}                    &  \mr{1}{*}{$\Sat$}    &  0  &  0.00  &    5,993   &   784,681.66   &   566   &    16,827.79   \\ \cline{2-8}
                                  &  \mr{1}{*}{$\Unsat$}  &  0  &  0.00  &    2,249   &   504,022.31   &\b{504}  & \b{17,919.88}  \\ \hline\hline
\mr{2}{*}{\barcelogicMaxsmt}      &  \mr{1}{*}{$\Sat$}    &  0  &  0.00  &   11,321   &   262,793.96   &   895   &     6,530.07   \\ \cline{2-8}
                                  &  \mr{1}{*}{$\Unsat$}  &  0  &  0.00  &    2,618   &    35,838.06   &   148   &    14,481.72   \\ \hline\hline
\mr{2}{*}{\barcelogicJmpNoCores}  &  \mr{1}{*}{$\Sat$}    &  0  &  0.00  &\b{11,522} & \b{246,918.87}  &\b{943}  & \b{15,074.33}  \\ \cline{2-8}
                                  &  \mr{1}{*}{$\Unsat$}  &  0  &  0.00  &    2,502   &    51,699.62   &   129   &     1,722.72   \\ \hline\hline
\mr{2}{*}{\barcelogicJmpCores}    &  \mr{1}{*}{$\Sat$}    &  0  &  0.00  &   11,504   &   244,201.03   &   941   &    12,174.59   \\ \cline{2-8}
                                  &  \mr{1}{*}{$\Unsat$}  &  0  &  0.00  &    2,573   &    49,572.32   &   142   &     7,394.77   \\ \hline
  \end{tabular}
}
\end{table}
\end{center}

These results confirm that, in general, our techniques work well on
$\Sat$ instances: except for families \textsf{leipzig}, \textsf{mcm}
and \textsf{UA}, the best tool is one of the \textsf{bcl-*} solvers.
The gap with respect to
\yices2 and \z3 is particularly remarkable on benchmarks coming from
our termination proving applications (families \textsf{Term},
\textsf{Integer Transition Systems} and \textsf{C Integer}).

On the other hand, as was already justified in Section
\ref{sec:smt_nia_experiments_all}, regarding $\Unsat$ problems, in
some families the \textsf{bcl-*} solvers are clearly outperformed by
the CAD-based techniques of \yices2 and \z3. This suggests that a
mixed approach that used our methods as a filter and that fell
back to CAD after some time threshold could possibly take the best of
both worlds.

Comparing our techniques among themselves, there is not an overall
clear winner. For $\Sat$ examples, it can be seen that the
non-incremental approach is indeed a useful heuristic:
\barcelogicJmpNoCores tends to perform better, being the biggest
difference in the \textsf{Integer Transition Systems} and \textsf{C
  Integer} families. As regards optimality cores, as could be expected
on $\Sat$ instances they do not prove profitable and result into a
slight overhead of \barcelogicJmpCores with respect to
\barcelogicJmpNoCores. On the other hand, on $\Unsat$ examples quite
often (namely, families \textsf{Term}, \textsf{LassoRanker},
\textsf{Integer Transition Systems} and \textsf{C Integer}) the
shorter conflict clauses discarding previous combinations of
artificial domains help in detecting unsatisfiability more
efficiently. Still, for this kind of instances \barcelogicMaxsmt is
usually the best of the three, since fewer iterations of the loop in
procedure \solveSMTQFNIAMinModels are required to prove that the
formula is unsatisfiable.

\section{Solving \MAXSMT(QF-NIA)}
\label{sec:maxsmt_nia}

This section is devoted to the extension of our techniques for
\SMT(QF-NIA) to \MAXSMT(QF-NIA), which has a wide range of
applications, e.g. in termination and non-termination proving
\cite{DBLP:conf/fmcad/LarrazORR13,DBLP:conf/cav/LarrazNORR14} as well
as safety analysis \cite{DBLP:conf/fmcad/BrockschmidtLOR15}. Taking
into account the results of the experiments in Sections
\ref{sec:smt_nia_experiments_all} and
\ref{sec:smt_nia_experiments_maxsmt}, we will choose the
\MAXSMT(QF-LIA) approaches as \SMT(QF-NIA) solving engines for the
rest of this article. In particular, in the description of the
following algorithms we will take as a reference the first version
explained in Section \ref{sec:maxsmt_la}, since adapting the
algorithms to the variations from Sections \ref{sec:non-incremental}
and \ref{sec:optimality-cores} is easy.

\subsection{Algorithm}
\label{sec:maxsmt_nia_algorithm}

We will represent the input $F_0$ of a \MAXSMT(QF-NIA) instance as a
conjunction of a set of hard clauses
$\hard_0 = \{ C_1, \cdots, C_n \}$ and a set of soft clauses
$\soft_0 = \{ [D_1, \Weight_1], \cdots, [D_m, \Weight_m]\}$. The aim
is to decide whether there exist assignments $\alpha$ such that
$\alpha \models \hard_0$, and if so, to find one such that
$\sum\{ \Weight \st {[D, \Weight] \in \soft_0, \alpha \not \models
  D}\}$ is minimized.

\begin{algorithm}[t]
\begin{lstlisting}
$\langle$status, map$\rangle$ $\solveMaxSMTQFNIA$(formula $F_0$) { // returns if $H_0$ is satisfiable and best model wrt. $S_0$
                                                         // $H_0$ are the hard clauses of $F_0$ and $S_0$ the soft ones
  $\bounds$ = artificial_bounds($F_0$);
  $\langle H,S \rangle$ = linearize($F_0$, $\bounds$);
  $\bsf$ = $\undef$;                                                  // best model found so far
  $\msc$ = $\infty$;                                                  // maximum soft cost we can afford
  while (not timed_out()) {
    $\langle \status, \model \rangle$ = $\optimizeQFLIAMaxSMTthreshold$($H$, $S$, $B$, $\msc$);
    if ($\status$ == $\Unsat$)
      if ($\bsf$ == $\undef$) return $\langle\Unsat, \undef\rangle$;
      else                  return $\langle\Sat, \bsf\rangle$;
    else if ($\cost_\bounds(\model)$ == $0$) {
      $\bsf$ = $\model$ ;
      $\msc$ = $\cost_\soft(\model) - 1$;                                // let us assume costs are natural numbers
    }
    else {
      $\bounds'$ = $\relaxDomainsMinModels$($\bounds$, $\model$);
      $H$ = $\update$($H$, $\bounds$, $\bounds'$);                                        // add case splitting clauses to the hard part
      $\bounds$ = $\bounds'$;
    }
  }
  return $\langle\Unknown, \undef\rangle$;
}
\end{lstlisting}
\caption{Algorithm for solving \MAXSMT(QF-NIA)}
\label{alg:maxsmt_nia}
\end{algorithm}

\begin{algorithm}[t]
\begin{lstlisting}
$\langle$status, map$\rangle$ $\optimizeQFLIAMaxSMTthreshold$(formula $H$, formula $S$, set $\bounds$, number $\mathit{msc}$) {
  $F' = H \cup S$;                                               // $H$ are hard clauses, $S$ are soft
  for ($V \geq L$ in $\bounds$)
    $F'$ = $F'$ $\cup$ $\{ [V \geq L, 1]\}$;                                     // added as a soft clause, e.g. with weight 1
  for ($V \leq U$ in $\bounds$)
    $F'$ = $F'$ $\cup$ $\{ [V \leq U, 1]\}$;                                     // added as a soft clause, e.g. with weight 1

  return $\solveMaxSMTQFLIA$($F'$, $\mathit{msc}$);        // call to \maxsmt solver
}
\end{lstlisting}
\caption{Procedure $\optimizeQFLIAMaxSMTthreshold$}
\label{alg:optimize_QF_LIA:maxsmt_threshold}
\end{algorithm}

Procedure $\solveMaxSMTQFNIA$ for solving \MAXSMT(QF-NIA) is shown in
Algorithm \ref{alg:maxsmt_nia}. In its first step, as usual the initial
artificial bounds $\bounds$\footnote{We will abuse notation and
  represent with $\bounds$ both the set of artificial bounds and also
  the corresponding set of weighted clauses. The exact meaning will be
  clear from the context.} are chosen (procedure \artificialBounds),
with which the input formula $F_0 \equiv \hard_0 \land \soft_0$ is
linearized (procedure \linearize). As a result, a weighted linear
formula is obtained with hard clauses $\hard$ and soft clauses
$\soft$, where:

\begin{itemize}

\item $\hard$ results from replacing the non-linear monomials in
  $\hard_0$ by their corresponding fresh variables, and adding the
  case splitting clauses;

\item $\soft$ results from replacing the non-linear monomials in
  $\soft_0$ by their corresponding fresh variables.

\end{itemize}

Now notice that there are two kinds of weights: those from the
original soft clauses, and those introduced in the linearization. As
they have different meanings, it is convenient to consider them
separately. Thus, given an assignment $\alpha$, we define its
\emph{(total) cost} as
$\cost(\alpha) = (\cost_\bounds(\alpha), \cost_\soft(\alpha))$, where
$\cost_\bounds(\alpha) = \sum\{\weight \st [b, \weight] \in \bounds, \alpha \not \models b\}$
is the \emph{bound cost}, i.e., the contribution to the total cost due to artificial bounds,
and
$\cost_\soft(\alpha) = \sum\{\Weight \st [D, \Weight] \in \soft, \alpha \not \models D\}$
is the \emph{soft cost}, corresponding to the original soft clauses.
Equivalently, if weights are written as pairs, so that artificial
bound clauses become of the form $[C, (\weight, 0)]$ and soft clauses
become of the form $[C, (0, \Weight)]$, we can write
$\cost(\alpha) = \sum\{ (\weight, \Weight) \st [C, (\weight, \Weight)] \in S \cup \bounds, \alpha \not \models C\} ,$
where the sum of the pairs is component-wise.
In what follows, pairs $(\cost_\bounds(\alpha), \cost_\soft(\alpha))$
will be lexicographically compared, so that the bound cost (which
measures the consistency with respect to the theory of QF-NIA) is more relevant than the soft
cost. Hence, by taking this cost function and this ordering, we have a
\MAXSMT(QF-LIA) instance in which weights are not natural or
non-negative real numbers, but pairs of them.

In the following step of the algorithm $\solveMaxSMTQFNIA$, the procedure
$\optimizeQFLIAMaxSMTthreshold$ (described in Algorithm
\ref{alg:optimize_QF_LIA:maxsmt_threshold}) dispatches this \maxsmt
instance. This procedure is like that presented in Algorithm
\ref{alg:optimize_QF_LIA:maxsmt}, with the only difference that now a
parameter $\msc$ is passed to the \MAXSMT(QF-LIA) solver. This
parameter restrains the models of the hard clauses the solver will
consider: only assignments $\alpha$ such that
$\cost_\soft(\alpha) \leq \msc$ will be taken into account. That is,
this adapted \maxsmt solver computes, among the models $\alpha$ of the
hard clauses such that $\cost_\soft(\alpha) \leq \msc$ (if any), one
that minimizes $\cost(\alpha)$. Thus, the search can be pruned when it
is detected that it is not possible to improve the best soft cost
found so far. This adjustment is not difficult to implement if the
\maxsmt solver follows a branch-and-bound scheme (see Section
\ref{sec:maxsmt_la}), as it is our case.

Then the algorithm examines the result of the call to the \maxsmt
solver. If it is $\Unsat$, then there are no models of the hard
clauses with soft cost at most $\msc$. Therefore, the algorithm can
stop and report the best solution found so far, if any.

Otherwise, $\model$ satisfies the hard clauses and has soft cost at
most $\msc$. If it has null bound cost, and hence is a true model of
the hard clauses of the original formula, then the best solution found
so far and $\msc$ are updated, in order to search for a solution with
better soft cost. Finally, if the bound cost is not null, then domains
are relaxed as described in Section \ref{sec:maxsmt_la}, in order to
widen the search space. In any case, the algorithm jumps back and a
new iteration is performed.

For the sake of simplicity, Algorithm \ref{alg:maxsmt_nia} returns
$\langle\Unknown, \undef\rangle$ when time is exhausted. However, the
best model found so far $\bsf$ can also be reported, as it can still
be useful in practice.

The following theorem states the correctness of procedure
$\solveMaxSMTQFNIA$:

\begin{theorem}
  Procedure $\solveMaxSMTQFNIA$ is correct. That is, given a weighted formula
  $F_0$ in QF-NIA with hard clauses $H_0$ and soft clauses $S_0$:

  \begin{enumerate}

  \item if $\solveMaxSMTQFNIA$($F_0$) returns
    $\langle\Sat,\model\rangle$ then $H_0$ is satisfiable, and
    $\model$ is a model of $H_0$ that minimizes the sum of the weights
    of the falsified clauses in $S_0$; and

    \smallskip

  \item if $\solveMaxSMTQFNIA$($F_0$) returns
    $\langle\Unsat,\undef\rangle$ then $H_0$ is unsatisfiable.

  \end{enumerate}

\end{theorem}

\begin{proof}

    Let us assume that $\solveMaxSMTQFNIA$($F_0$) returns
    $\langle\Sat,\model\rangle$. The assignment $\model$ is different
    from $\undef$, and hence has been previously computed in a call to
    $\optimizeQFLIAMaxSMTthreshold$($H$, $S$, $B$, $\msc$) such that
    $\cost_B(M) = 0$. So $\model$ respects all artificial bounds in $B$.
    Thanks to the case splitting clauses in $H$, this ensures that
    auxiliary variables representing non-linear monomials have the right
    values. Therefore $\model$ satisfies $H_0$, which is what we wanted
    to prove. Now we just need to check that indeed $\model$ minimizes
    the sum of the weights of the falsified clauses in $S_0$. Notice
    that, from the specification of $\optimizeQFLIAMaxSMTthreshold$, we
    know that there is no model of $H$ such that its soft cost is
    strictly less than $\cost_\soft(\model)$. Now let $\model'$ be a
    model of $H_0$. By extending $\model'$ so that auxiliary variables
    representing non-linear monomials are assigned to their
    corresponding values, we have $\model' \models H$. By the previous
    observation,
    $\cost_{\soft_0}(\model') = \cost_\soft(\model') \geq
    \cost_\soft(\model) = \cost_{\soft_0}(\model)$.

    Now let us assume that $\solveMaxSMTQFNIA$($F_0$) returns
    $\langle\Unsat,\undef\rangle$. Let us also assume that there
    exists $\model'$ a model of $H_0$, and we will get a
    contradiction. Indeed, again extending $\model'$ as necessary, we
    have that $\model' \models H$. If $\solveMaxSMTQFNIA$($F_0$)
    returns $\langle\Unsat,\undef\rangle$, then the previous call to
    $\optimizeQFLIAMaxSMTthreshold$($H$, $S$, $B$, $\msc$) has returned
    $\langle\Unsat, \undef\rangle$, and moreover no previous call to
    $\optimizeQFLIAMaxSMTthreshold$ has produced a model with null
    bound cost. This means that $\msc$ has not changed its initial
    value, namely $\infty$. Therefore $H$ must be unsatisfiable, a
    contradiction.
\end{proof}

\begin{example}
Let $F_0$ be the weighted formula with hard clauses
$$H_0 \equiv tx + y \,\geq\, 4 \,\;\land\;\, t^2 + x^2 + w^2 + y^2 \,\leq\, 12 $$
(the same of previous examples) and a single soft clause
$$S_0 \equiv [ t^2 + x^2 + y^2 \leq 1,\; 1].$$

\noindent
Let us take $-1 \leq t, x, w, y \leq 1$ as artificial bounds.
After linearization, we get a weighted linear formula with hard clauses:

\begin{center}

  $
  H \equiv \left(
\begin{array}{ll}
v_{tx} + y \geq 4 \;\land\; v_{t^2} + v_{x^2} + v_{w^2} + v_{y^2} \leq 12 \;\;\land\medskip\\
\begin{array}{lcl c lcl c}
( t =-1 & \limplies & v_{tx} = -x   ) & \;\;\;\land\;\;\; &  \\
( t = 0 & \limplies & v_{tx} =  0   ) & \;\;\;\land\;\;\; &  \\
( t = 1 & \limplies & v_{tx} =  x   ) & \;\;\;\land\;\;\; &  \\
\\
( t =-1 & \limplies & v_{t^2} =  1  ) & \;\;\;\land\;\;\; &  ( w =-1 & \limplies & v_{w^2} =  1  ) & \;\;\land\\
( t = 0 & \limplies & v_{t^2} =  0  ) & \;\;\;\land\;\;\; &  ( w = 0 & \limplies & v_{w^2} =  0  ) & \;\;\land\\
( t = 1 & \limplies & v_{t^2} =  1  ) & \;\;\;\land\;\;\; &  ( w = 1 & \limplies & v_{w^2} =  1  ) & \;\;\land\\
\\
( x =-1 & \limplies & v_{x^2} =  1  ) & \;\;\;\land\;\;\; &  ( y =-1 & \limplies & v_{y^2} =  1  ) & \;\;\land\\
( x = 0 & \limplies & v_{x^2} =  0  ) & \;\;\;\land\;\;\; &  ( y = 0 & \limplies & v_{y^2} =  0  ) & \;\;\land\\
( x = 1 & \limplies & v_{x^2} =  1  ) & \;\;\;\land\;\;\; &  ( y = 1 & \limplies & v_{y^2} =  1  ) & \;\;\land\\
\end{array}
\end{array}
\right)
$
\end{center}

\noindent
and soft clauses

\begin{center}
  $
  \begin{array}{l}
S \equiv [v_{t^2} + v_{x^2} + v_{y^2} \leq 1,\; (0, 1)] \medskip \\
B \equiv \left(\begin{array}{rcrc} \\
    \left[-1 \leq t,\; (1, 0)\right]  & \;\land\; &  \left[ t \leq 1,\; (1, 0)\right] & \;\land \smallskip\\
    \left[-1 \leq x,\; (1, 0)\right]  & \;\land\; &  \left[ x \leq 1,\; (1, 0)\right] & \;\land \smallskip\\
    \left[-1 \leq w,\; (1, 0)\right]  & \;\land\; &  \left[ w \leq 1,\; (1, 0)\right] & \;\land \smallskip\\
    \left[-1 \leq y,\; (1, 0)\right]  & \;\land\; &  \left[ y \leq 1,\; (1, 0)\right] &
  \end{array}\right),
  \end{array}
  $
\end{center}
where weights are already represented as pairs
$(\mathrm{bound\ cost}, \mathrm{soft\ cost})$ as explained above.

In the first call to $\optimizeQFLIAMaxSMTthreshold$($H$, $S$, $B$,
$\infty$), the optimal cost is $(1, 0)$. An assignment with this cost
that may be returned is, for example, $t = v_{t^2} = 1$,
$x = v_{tx} = 4$ and $w = v_{w^2} = y = v_{y^2} = v_{x^2} = 0$, the
same in as Example \ref{ex:maxsmt}. In this assignment, the only soft
clause that is violated is $[ x \leq 1, (1, 0)]$.

Since the bound cost is not null, new artificial bounds should be
introduced. Following Example \ref{ex:maxsmt}, the new upper bound for
$x$ becomes $x \leq 4$. Hence, the soft clause
$[ x \leq 1, (1, 0)]$ is replaced by $[ x \leq 4, (1, 0)]$), and the
following hard clauses are added:

\begin{center}
$
\begin{array}{lcl}
x = 2 & \limplies & v_{x^2} = 4\\
x = 3 & \limplies & v_{x^2} = 9\\
x = 4 & \limplies & v_{x^2} = 16\\
\end{array}
$
\end{center}

The following call to $\optimizeQFLIAMaxSMTthreshold$ returns an
assignment with cost $(0, 1)$, e.g.,
$t = v_{t^2} = w = v_{w^2} = y = v_{y^2} = 1$, $x = v_{tx} = 3$, $1$
and $v_{x^2} = 9$. Since the bound cost is null, this assignment is
recorded as the best model found so far and $\msc$ is set to $0$.
This forces that, from now on, only solutions with null soft cost are
considered, i.e., the soft clause $v_{t^2} + v_{x^2} + v_{y^2} \leq 1$
must hold. Since $t^2 + x^2 + y^2 \leq 1$ implies
$|t|, |x|, |y| \leq 1$, which contradicts $tx + y \geq 4$, there is
actually no solution of cost $(0, 0)$. Hence next calls to
$\optimizeQFLIAMaxSMTthreshold$ will unsuccessfully look for
non-linear models with null soft cost. If additional case splitting
clauses considering values outside the finite domain are introduced,
such as for instance
\begin{center}
$
\begin{array}{lcl}
x \leq -2 & \limplies & v_{x^2} \geq 4 \\
x \geq 5 & \limplies & v_{x^2} \geq 25\,, \\
\end{array}
$
\end{center}
unsatisfiability will be detected by $\optimizeQFLIAMaxSMTthreshold$.
Then procedure $\solveMaxSMTQFNIA$ will terminate reporting that the
minimum cost (with respect to the original soft clauses $S_0$) is $1$,
and that a model with that cost is given by $t = w = y = 1$ and
$x = 3$. \hfill$\blacksquare$
\end{example}

\subsection{Experimental Evaluation}
\label{sec:maxsmt_nia_experiments}

In this section we evaluate experimentally the approach proposed in
Section \ref{sec:maxsmt_nia_algorithm} for solving \MAXSMT(QF-NIA). We
adapt the method to each of the three \MAXSMT(QF-LIA)-based variants
for solving \SMT(QF-NIA). Thus, following the same names as in Section
\ref{sec:smt_nia_experiments_maxsmt}, here we consider the solvers
\barcelogicMaxsmt, \barcelogicJmpNoCores and \barcelogicJmpCores. We
also include in the experiments \z3, which is the only competing tool
that, up to our knowledge, can handle \MAXSMT(QF-NIA) too. As regards
benchmarks, we use the original \MAXSMT(QF-NIA) versions (that is,
keeping soft clauses) of the examples \textsf{Integer Transition
  Systems} and \textsf{C Integer} employed in Section
\ref{sec:smt_nia_experiments_maxsmt}.

\begin{center}
  \begin{table}%

  \tbl{Experimental evaluation of \MAXSMT(QF-NIA) solvers on benchmark family \textsf{Integer Transition
  Systems} (20354 benchmarks).\label{tab:maxsmt_nia_experiments_its}}{%

  \begin{tabular}{|c|r|r||r|r||r|r||r|r|}
   \cline{2-9}

  \multicolumn{1}{c|}{} & \multicolumn{2}{c||}{\barcelogicMaxsmt}   & \multicolumn{2}{c||}{\barcelogicJmpNoCores} & \multicolumn{2}{c||}{\barcelogicJmpCores} & \multicolumn{2}{c|}{\z3} \\ \cline{2-9}
  \multicolumn{1}{c|}{} & \# p. & time & \# p. & time & \# p. & time & \# p. & time \\ \hline

$\Unsat$        & \b{2,618}  &  \b{32,947.80} & 2,490 &  49,351.13 & 2,573 &  46,750.26 & 2,571 & 624,204.13 \\ \hline
$\Optimal$      & \b{7,644}  & \b{449,806.47} & 6,720 & 174,062.35 & 6,908 & 228,974.31 &     0 &       0.00 \\ \hline
$\Optimal+\Sat$ & \b{8,311}  & \b{490,204.00} & 7,390 & 218,202.00 & 7,583 & 276,237.00 & 2,165 & 652,295.00 \\ \hline
  \end{tabular}
}
\end{table}
\end{center}

\begin{center}
  \begin{table}%

  \tbl{Experimental evaluation of \MAXSMT(QF-NIA) solvers on benchmark family \textsf{C Integer} (2019 benchmarks).\label{tab:maxsmt_nia_experiments_cint}}{%

  \begin{tabular}{|c|r|r||r|r||r|r||r|r|}
   \cline{2-9}

  \multicolumn{1}{c|}{} & \multicolumn{2}{c||}{\barcelogicMaxsmt}   & \multicolumn{2}{c||}{\barcelogicJmpNoCores} & \multicolumn{2}{c||}{\barcelogicJmpCores} & \multicolumn{2}{c|}{\z3} \\ \cline{2-9}
  \multicolumn{1}{c|}{} & \# p. & time & \# p. & time & \# p. & time & \# p. & time \\ \hline

$\Unsat$        &   144 &   9,027.27 &   121 &   3,993.28 &   136  &      8,930.16  &\b{257} & \b{7,855.24} \\ \hline
$\Optimal$      &   453 &   9,090.26 &   466 &   9,177.07 &\b{469} &  \b{10,768.57} &     0  &        0.00 \\ \hline
$\Optimal+\Sat$ &   522 &   9,108.00 &   535 &   9,194.00 &\b{539} &  \b{10,797.00} &   207  &   23,579.00 \\ \hline

  \end{tabular}
}
\end{table}
\end{center}

Tables \ref{tab:maxsmt_nia_experiments_its} and
\ref{tab:maxsmt_nia_experiments_cint} show the results of the
experiments on the families \textsf{Integer Transition Systems} and
\textsf{C Integer}, respectively. In each table, row $\Unsat$
indicates the number of instances that were proved to be
unsatisfiable, and row $\Optimal$ counts the instances for which
optimality of the reported model could be established. A third row
$\Optimal+\Sat$ adds to row $\Optimal$ the number of problems in which
a model was found, but could not be proved to be optimal. For the sake
of succinctness, as in previous tables other outcomes (timeouts,
$\Unknown$ answer, etc.) are not made explicit. Columns represent
systems and show either the number of problems that were solved with a
certain outcome, or the total time (in seconds) to process all of
them. The best solver in each case is highlighted in bold face.

From the tables it can be observed that \barcelogicJmpCores is more
effective than \barcelogicJmpNoCores for \MAXSMT. This is natural:
proving the optimality of the best model found so far implicitly
involves proving unsatisfiability, more precisely that there cannot be
a model with a better cost. And as was already remarked in Section
\ref{sec:smt_nia_experiments_maxsmt}, optimality cores help the
non-incremental approach to detect unsatisfiability more quickly.
Regarding the incremental approach, the results are inconclusive:
depending on the benchmarks, \barcelogicMaxsmt may perform better than
\barcelogicJmpCores, or the other way around. Finally, \z3 is
competitive or even superior when dealing with unsatisfiable problems,
while it significantly lags behind for the rest of the instances.

\section{Solving \smt  and \MAXSMT(\ea)}
\label{sec:maxsmt_ea}

In this section we will extend our techniques for \smt and
\MAXSMT(QF-NIA) to the theory of \ea. In this fragment of first-order
logic, formulas are of the form $\exists x \, \forall y \, F(x,y)$,
where $F$ is a quantifier-free formula whose literals are polynomial
inequalities. Moreover, the existentially quantified variables have
integer type, whereas the universally quantified ones are real. In
particular we will focus on a subset of this logic, namely, those
formulas in which monomials never contain the product of two
universally quantified variables.

This fragment of quantified non-linear arithmetic is relevant to many applications.
For example, it
appears in verification and synthesis problems
when the so-called \emph{template-based method}
\cite{ColonSipma2002CAV} is employed. In this framework, one attempts
to discover an object of interest (e.g., an invariant, or a ranking
function) by introducing a \emph{template}, usually a linear
inequality or expression, and solving a formula that represents the
conditions the object should meet. For instance, let us find an
invariant for the next loop:

\begin{center}
@double y = 0; while (y <= 2) y = y+1;@
\end{center}

\noindent
A loop invariant $I(y)$ must satisfy the following \emph{initiation}
and \emph{inductiveness} conditions:

\begin{itemize}
\item \textbf{Initiation:}\hspace*{1.05cm} $\forall y_0 \;\; (y_0 = 0 \;\limplies\; I(y_0))$
  \smallskip
\item \textbf{Inductiveness:}\quad $\forall y_1, y_2 \;\; (I(y_1) \;\land\; y_1 \leq 2 \;\land\; y_2 = y_1+1 \;\,\limplies\;\, I(y_2))$
\end{itemize}

Now a linear template $x_0\, y \leq x_1$ is introduced as a
candidate for $I(y)$, where $x_0$, $x_1$ are unknowns and $y$ is the
program variable. Then the conditions needed for $I(y)$ to be an
invariant can be expressed in terms of template unknowns and program
variables as an $\exists\forall$ formula:
$$
\begin{array}{rll}
\exists x_0, x_1\;\, \forall y_0, y_1, y_2\, & \big( (y_0 = 0 \;\limplies\; x_0\, y_0 \leq x_1) \;\land \\
& (x_0\, y_1 \leq x_1 \;\land\; y_1 \leq 2 \;\land\; y_2 = y_1+1 \,\;\limplies\;\, x_0\, y_2 \leq x_1) \big)\\
\end{array}
$$
This falls into the logical fragment considered here. Indeed note
that, since the template is linear, the non-linear monomials in the
formula always consist of the product of a template unknown and a
program variable. Moreover, we can regard that we are interested in
integer coefficients, so the existential variables are integers, while
the universal variables are reals, since the program variable is a
\texttt{double}. On the other hand, if one is interested in finding
models with other type patterns, the following can be taken into
account: in general, if a formula
$$\exists x \in \integers \;\; \forall y \in \reals \;\; F(x,y)$$ is
satisfiable, then so are
\begin{itemize}
\item $\exists x \in \reals \;\; \forall y \in \reals \;\; F(x,y),$
\item $\exists x \in \integers \;\; \forall y \in \integers \;\; F(x,y),$
\item $\exists x \in \reals \;\; \forall y \in \integers \;\; F(x,y),$
\end{itemize}
since the same witness $x$ can be taken.

\subsection{Algorithm}
\label{sec:maxsmt_ea_algorithm}

Let us first describe how to deal with the satisfiability problem given
a formula $\exists x \, \forall y \, F(x,y)$, and then the technique
will extend to the more general \MAXSMT(\ea) problem naturally. Note
that the requirement that monomials cannot contain the product of two
universal variables allows writing the literals in $F$ as linear
polynomials in variables $y$, i.e., in the form
$a_{1}(x)\, y_1 + \cdots + a_{n}(x)\, y_n \leq b(x)$. Hence, if for
instance $F$ is a clause, we can write it as
$$\neg\Big(\bigwedge_{i=1}^{m} a_{i1}(x)\, y_1 + \cdots + a_{in}(x)\, y_n \leq b_i(x) \;\land\; \bigwedge_{j=1}^{l} c_{j1}(x)\, y_1 + \cdots + c_{jn}(x)\, y_n < d_j(x)\Big),$$
or more compactly using matrix notation as $\neg\Big(A(x)\, y \leq b(x) \;\land\; C(x)\, y < d(x)\Big)$. 

The key idea (borrowed from \cite{ColonSipma2002CAV}\footnote{In 
\cite{ColonSipma2002CAV}, Farkas' Lemma is used instead of the generalization 
presented here.}) is to apply the
following result from polyhedral geometry to eliminate the quantifier
alternation and transform the problem into a purely existential one:

\begin{theorem} [Motzkin's Transposition Theorem \cite{Schrijver}]
Let $A \in \reals^{m \times n}$, $y \in \reals^n$, $b \in \reals^m$,
$C \in \reals^{l \times n}$ and $d \in \reals^l$. The system
$A(x)\, y \leq b(x)$ $\,\land\, C(x)\, y < d(x)$ is unsatisfiable 
if and only if there are $\lambda \in \reals^m$ and $\mu \in \reals^l$
such that $\lambda \geq 0$, $\mu \geq 0$, $\lambda^T A(x) + \mu^T C(x) = 0$,
$\lambda^T b(x) + \mu^T d(x) \leq 0$, and either $\lambda^T b(x) < 0$ or $\mu \neq 0$.
\end{theorem}

Thanks to Motzkin's Transposition Theorem, we have that formulas
$$\exists x \; \forall y \; \neg\big(A(x)\, y \leq b(x) \;\land\; C(x)\, y < d(x))$$ and
$$\exists x \, \exists \lambda \, \exists \mu \,\Big( \lambda,\mu \geq 0 
\;\land\; \lambda^T A(x) + \mu^T C(x) = 0 \;\land\; \lambda^T b(x) + \mu^T d(x) \leq 0 \;\land\;
(\lambda^T b(x) < 0 \;\lor\; \mu \neq 0)\Big)$$
are equisatisfiable. In general, if the formula $F$ in
$\exists x \, \forall y \, F(x,y)$ is a CNF, this transformation is
applied locally to each of the clauses with fresh multipliers.

Note that the formula resulting from applying Motzkin's Transposition Theorem is
non-linear, but the existentially quantified variables $\lambda$ and $\mu$ have
real type. Fortunately, our techniques from Section
\ref{sec:domain_relaxation} do not actually require that all variables
are integer: it suffices that there are \emph{enough} finite domain
variables to perform the linearization. And this is indeed the case,
since every non-linear monomial of the transformed formula has at most
one occurrence of a $\lambda$ or a $\mu$ variable, and all other variables are
integer. All in all, we have reduced the problem of satisfiability of
the fragment of \ea under consideration to satisfiability of
non-linear formulas that our approach can deal with.

Finally, regarding \maxsmt the technique extends clause-wise in a
natural way. Given a weighted CNF, hard clauses are transformed using
Motzkin's Transposition Theorem as in the \smt case. As for soft clauses, let
$[S, \Weight]$ be such a clause, where $S$ is of the form
$\neg(A(x)\, y \leq b(x) \,\land\, C(x)\, y < d(x))$. Then a fresh
propositional symbol $p_S$ is introduced, and $[S, \Weight]$ is
replaced by a soft clause $[p_S, \Weight]$ and hard clauses
corresponding to the double implication
$$\Big( \lambda,\mu \geq 0 
\;\land\; \lambda^T A(x) + \mu^T C(x) = 0 \;\land\; \lambda^T b(x) + \mu^T d(x) \leq 0 \;\land\;
(\lambda^T b(x) < 0 \;\lor\; \mu \neq 0)\Big)
\leftrightarrow p_S.$$
Therefore, similarly to satisfiability, we can solve the \maxsmt
problem for the fragment of \ea of interest by reducing it to
instances that can be handled with the techniques presented in Section
\ref{sec:maxsmt_nia}.

\begin{example}
  \label{ex:ea}
  Let us consider again the problem of finding an invariant for the
  loop:

\begin{center}
@double y = 0; while (y <= 2) y = y+1;@
\end{center}

However, now we will make the initiation condition soft, say with
weight $1$, while the inductiveness condition will remain hard (as
done in \cite{DBLP:conf/fmcad/BrockschmidtLOR15}). The rationale is
that, if the initiation condition can be satisfied, then we have a
true invariant; and if it is not, then at least we have a
\emph{conditional invariant}: a property that, if at some iteration
holds, then from that iteration on it always holds.

Using the same template as above, the formula to be solved is
(quantifiers are omitted for the sake of presentation):
$$
\begin{array}{ll}
[y_0 = 0 \;\limplies\; x_0\, y_0 \leq x_1, \; 1] \;\land \\
(x_0\, y_1 \leq x_1 \;\land\; y_1 \leq 2 \;\land\; y_2 = y_1+1 \,\;\limplies\;\, x_0\, y_2 \leq x_1) \\
\end{array}
$$
After moving the right-hand side of the implication to the left, and applying some simplifications, it results into:
$$
\begin{array}{ll}
[0 \leq x_1, \; 1] \;\land \\
\neg(x_0\, y_1 \leq x_1 \;\land\; y_1 \leq 2 \,\;\land\;\, x_0\, (y_1+1) > x_1) \\
\end{array}
$$
Now the transformation is performed clause by clause. Since the first
clause $[0 \leq x_1, \; 1]$ does no longer contain universally
quantified variables, it can be left as it is. As regards the second
one, we introduce three fresh multipliers $\lambda_1$, $\lambda_2$, and $\mu$
and replace
$$\neg(x_0\, y_1 \leq x_1 \;\land\; y_1 \leq 2 \,\;\land\;\, x_0\, (y_1+1) > x_1)$$
by
$$
\begin{array}{c}
\Big(\lambda_1 \geq 0 \land \lambda_2 \geq 0 \land \mu \geq 0 \land 
\lambda_1 x_0 + \lambda_2 - \mu x_0 = 0  \;\;\land\\
\lambda_1 x_1 + 2\lambda_2 + \mu (x_0 - x_1) \leq 0 \land 
(\lambda_1 x_1 + 2\lambda_2 < 0 \lor \mu \neq 0)\Big)
\end{array}
$$
All in all, the following \maxsmt instance must be solved:
$$
\begin{array}{c}
[0 \leq x_1, \; 1] \;\;\land
\smallskip\\
\Big(\lambda_1 \geq 0 \land \lambda_2 \geq 0 \land \mu \geq 0 \land 
\lambda_1 x_0 + \lambda_2 - \mu x_0 = 0  \;\;\land\\
\lambda_1 x_1 + 2\lambda_2 + \mu (x_0 - x_1) \leq 0 \land 
(\lambda_1 x_1 + 2\lambda_2 < 0 \lor \mu \neq 0)\Big)
\end{array}
$$
There exist many solutions with cost $0$, each of them corresponding
to a loop invariant; for instance,
$x_0 = 1, x_1 = 3, \lambda_1 = 0, \lambda_2 = 1, \mu = 1$ (which represents the
invariant @y <= 3@). \hfill$\blacksquare$
\end{example}

\subsection{Experimental Evaluation}
\label{sec:maxsmt_ea_experiments}

In this section we evaluate experimentally the approach proposed in
Section \ref{sec:maxsmt_ea_algorithm} for solving \MAXSMT(\ea).
Similarly to Section \ref{sec:maxsmt_nia_experiments}, again we
instantiate the method for the three \MAXSMT(QF-LIA)-based variants
for solving \SMT(QF-NIA). So, using the same names as in Sections
\ref{sec:smt_nia_experiments_maxsmt} and
\ref{sec:maxsmt_nia_experiments}, in this evaluation we consider the
solvers \barcelogicMaxsmt, \barcelogicJmpNoCores and
\barcelogicJmpCores. Unfortunately, as far as we know no competing
tool can handle the problems of \MAXSMT(\ea) effectively. Hence, we
have to limit our experiments to our own tools.

Regarding benchmarks, again we use the weighted formulas of the
families \textsf{Integer Transition Systems} and \textsf{C Integer},
employed in Section \ref{sec:maxsmt_nia_experiments}. However, here
problems are expressed in \MAXSMT(\ea) rather than in \MAXSMT(NIA);
that is, Motzkin's Transposition Theorem is applied silently inside
the solver, and not in the process of generating the instances.
Moreover, as illustrated in Example \ref{ex:ea}, \MAXSMT(\ea) problems
coming from the application of the template-based method can usually
be simplified, e.g., by using equations to eliminate variables. In
order to introduce some variation with respect to the evaluation in
Section \ref{sec:maxsmt_nia_experiments}, we decided to experiment
with the \MAXSMT(\ea) problems in raw form, without simplifications.
Another difference is that, while in Section
\ref{sec:maxsmt_nia_experiments} multipliers were considered integer
variables (so that purely integer problems were obtained), in this
evaluation they have real type.

\begin{center}
  \begin{table}%

  \tbl{Experimental evaluation of \MAXSMT(\ea) solvers on benchmark family \textsf{Integer Transition
  Systems} (20354 benchmarks).\label{tab:maxsmt_ea_experiments_its}}{%

  \begin{tabular}{|c|r|r||r|r||r|r|}
   \cline{2-7}

  \multicolumn{1}{c|}{} & \multicolumn{2}{c||}{\barcelogicMaxsmt}   & \multicolumn{2}{c||}{\barcelogicJmpNoCores} & \multicolumn{2}{c||}{\barcelogicJmpCores} \\ \cline{2-7}
  \multicolumn{1}{c|}{} & \# p. & time & \# p. & time & \# p. & time \\ \hline

$\Unsat$         &\b{2,196} &    \b{89,259.58} &     2,119  &  121,556.46  &    2,031  &   140,585.27  \\ \hline
$\Optimal$       &\b{6,707} & \b{1,002,816.92} &     5,902  &  405,813.78  &    5,856  &   401,333.33  \\ \hline
$\Optimal+\Sat$  &\b{7,337} & \b{1,071,480.43} &     6,536  &  475,622.68  &    6,485  &   467,898.84  \\ \hline
  \end{tabular}
}
\end{table}
\end{center}

\begin{center}
  \begin{table}%

  \tbl{Experimental evaluation of \MAXSMT(\ea) solvers on benchmark family \textsf{C Integer} (2019 benchmarks).\label{tab:maxsmt_ea_experiments_cint}}{%

  \begin{tabular}{|c|r|r||r|r||r|r|}
   \cline{2-7}

  \multicolumn{1}{c|}{} & \multicolumn{2}{c||}{\barcelogicMaxsmt}   & \multicolumn{2}{c||}{\barcelogicJmpNoCores} & \multicolumn{2}{c||}{\barcelogicJmpCores} \\ \cline{2-7}
  \multicolumn{1}{c|}{} & \# p. & time & \# p. & time & \# p. & time \\ \hline
$\Unsat$         & \b{88} &    \b{10,095.79} &        64  &    1,992.78  &          76  &      2,545.89  \\ \hline
$\Optimal$       &   360  &       11,928.57  &       374  &   13,223.96  &      \b{379} &  \b{15,173.59} \\ \hline
$\Optimal+\Sat$  &   429  &       13,985.45  &       442  &   13,811.42  &      \b{447} &  \b{15,818.40} \\ \hline
  \end{tabular}
}
\end{table}
\end{center}

Results are shown in Tables \ref{tab:maxsmt_ea_experiments_its} and
\ref{tab:maxsmt_ea_experiments_cint}, following the same format as in
Section \ref{sec:maxsmt_nia_experiments}. It is worth noticing that
the number of solved instances is significantly smaller than in Tables
\ref{tab:maxsmt_nia_experiments_its} and
\ref{tab:maxsmt_nia_experiments_cint}, respectively. This shows the
usefulness of the simplifications performed when generating the
\MAXSMT(NIA) instances. Regarding which tool for \MAXSMT(\ea) among
the three is the most powerful, on $\Sat$ instances there is not a
global winner, while on unsatisfiable ones \barcelogicMaxsmt has the
best results for both families.

\section{Conclusions and Future Work}
\label{sec:conclusions}

In this article we have proposed two strategies to guide domain
relaxation in the instantiation-based approach for solving
\SMT(QF-NIA) \cite{Borrallerasetal2011JAR}. Both are based on
computing minimal models with respect to a cost function, namely: the
number of violated artificial domain bounds, and the distance with
respect to the artificial domains. We have experimentally argued that
the former gives better results than the latter and previous
techniques, and have devised further improvements, based on weakening
the invariant that artificial domains should grow monotonically, and
exploiting optimality cores. Finally, we have developed and
implemented algorithms for \MAXSMT(QF-NIA) and for \MAXSMT(\ea),
logical fragments with important applications to program analysis and
termination but which are missing effective tools.

As for future work, several directions for further research can be
considered. Regarding the algorithmics, it would be interesting to
look into different cost functions following the model-guided
framework proposed here, as well as alternative ways for computing
those minimal models (e.g., by means of \emph{minimal correction
  subsets}
\cite{DBLP:conf/ijcai/Marques-SilvaHJPB13,DBLP:conf/ijcai/BjornerN15}).
On the other hand, one of the shortcomings of our instantiation-based
approach for solving \MAXSMT/\SMT(QF-NIA) is that unsatisfiable
instances that require non-trivial non-linear reasoning cannot be
captured. In this context, the integration of real-goaled CAD
techniques adapted to SMT \cite{JovanovicMoura2012IJCAR} as a fallback
or run in parallel appears to be a promising line of work.

Another direction for future research concerns applications. So far we
have applied \MAXSMT/\SMT(QF-NIA/\ea) to array invariant generation
\cite{DBLP:conf/vmcai/LarrazRR13}, safety
\cite{DBLP:conf/fmcad/BrockschmidtLOR15}, termination
\cite{DBLP:conf/fmcad/LarrazORR13} and non-termination
\cite{DBLP:conf/cav/LarrazNORR14} proving. Other problems in program
analysis where we envision these techniques could help in improving
the state-of-the-art are, e.g., the analysis of worst-case execution
time and resource analysis. Also, so far we have only considered
sequential programs. The extension of \MAXSMT-based techniques to
concurrent programs is a promising line of work with a potentially
high impact in the industry.




\begin{acks}
  The authors would like to thank the anonymous referees of the
  conference version of this paper for their helpful comments.
\end{acks}

\bibliographystyle{ACM-Reference-Format-Journals}







\end{document}